\documentclass[11pt]{article}

\usepackage{geometry,float}
\usepackage{graphicx}
\usepackage{amssymb,mathtools}
\usepackage{amsmath}
\usepackage{amsthm,bm,bbm}
\usepackage{color}
\usepackage{tabularx}
\usepackage{enumitem}
\usepackage[authoryear]{natbib}
\usepackage{dsfont}
\usepackage{pifont}
\usepackage{algorithm}
\usepackage{algorithmic}
\usepackage{multirow}
\usepackage{JASA_manu}
\let\hat\widehat
\let\tilde\widetilde

\usepackage{url}
\usepackage[colorlinks,citecolor=blue,urlcolor=blue,linkcolor=blue,linktocpage=true]{hyperref}
\usepackage[nameinlink]{cleveref}
\crefformat{equation}{(#2#1#3)}
\crefrangeformat{equation}{(#3#1#4) to~(#5#2#6)}
\crefname{equation}{}{}
\Crefname{equation}{}{}

\newcommand{\single}{\renewcommand{\baselinestretch}{1.2}\normalsize}
\newcommand{\double}{\renewcommand{\baselinestretch}{1.63}\normalsize}

\global\long\def\data{\mathcal{D}}
\global\long\def\sdata{\mathcal{I}}
\global\long\def\alg{\mathcal{A}}

\global\long\def\t{\mathrm{T}}

\global\long\def\mean{\mathbb{E}}
\global\long\def\indc{\mathds{1}}

\renewenvironment{thebibliography}[1]{
	\begin{oldthebibliography}{#1}
		\setlength{\parskip}{0ex}
		\setlength{\itemsep}{0ex}
	}
	{
	\end{oldthebibliography}
}

\allowdisplaybreaks

\newtheorem{theorem}{Theorem}
\newtheorem{lemma}{Lemma}

\newtheorem{proposition}{Proposition}
\newtheorem{remark}{Remark}

\newtheorem{assumption}{Assumption}
\crefname{definition}{\textbf{definition}}{definitions}
\Crefname{definition}{Definition}{Definitions}
\crefname{assumption}{\textbf{assumption}}{assumptions}
\Crefname{assumption}{Assumption}{Assumptions}

\title{A Two-Sample Conditional Distribution Test Using Conformal Prediction and Weighted Rank Sum}
\author{Xiaoyu Hu \\
	School of Mathematical Sciences, Center for Statistical Science, Peking University, China\\
	and \\
	Jing Lei \\
	Department of Statistics and Data Science, Carnegie Mellon University, USA}

\begin{document}
	
	\begin{titlepage}
		\thispagestyle{empty}
		\single
		\maketitle
		
		\begin{center}
			\textbf{Abstract}
		\end{center}
		We consider the problem of testing the equality of conditional distributions of a response variable given a vector of covariates between two populations.  Such a hypothesis testing problem can be motivated from various machine learning and statistical inference scenarios, including transfer learning and causal predictive inference.
		We develop a nonparametric test procedure inspired from the conformal prediction framework.  The construction of our test statistic combines recent developments in conformal prediction with a novel choice of conformity score, resulting in a weighted rank-sum test statistic that is valid and powerful under general settings. To our knowledge, this is the first successful attempt of using conformal prediction for testing statistical hypotheses beyond exchangeability. Our method is suitable for modern machine learning scenarios where the data has high dimensionality and  large sample sizes, and can be effectively combined with existing classification algorithms to find good conformity score functions.  The performance of the proposed method is demonstrated in various numerical examples.
		
		\vspace*{.1in}
		
		
	\end{titlepage}
	
	\setcounter{page}{2}
	\double
	
	\section{Introduction}\label{sec:intro}
	
	Suppose we have two independent random samples
	$$ \{(X_{1i}, Y_{1i})\}_{i=1}^{n_1} \stackrel{i.i.d}{\sim} P_1 \quad \mathrm{and} \quad \{(X_{2j}, Y_{2j})\}_{j=1}^{n_2} \stackrel{i.i.d}{\sim} P_2 \,, $$
	where $P_1$, $P_2$ are distributions on a product space $\mathcal X\times\mathcal Y$.
	Here we consider a regression or classification setting where $X$ is the free variable (covariate), and $Y$ is the response variable.  We allow the spaces $\mathcal X$, $\mathcal Y$ to be general, and do not assume specific forms such as smoothness or parametric forms of $P_1$, $P_2$.  For example, $\mathcal X$ and $\mathcal Y$ can be multi-dimensional Euclidean spaces, manifolds, or discrete sets.   
	
	For $j=1,2$, let $P_j(\cdot|x)$ be the conditional distribution of $Y$ given $X=x$ under $P_j$, and $P_{j,X}(\cdot)$ be the corresponding marginal distribution of $X$.  We are interested in testing whether these two conditional distributions are the same.
	\begin{align}\label{eq:problem_i}
	H_0:~P_1( \cdot |x) = P_2( \cdot |x)~~\text{for all}~~x\in\mathcal X, \quad \mathrm{versus} \quad H_1:~~\text{otherwise}\,.
	\end{align}
	We illustrate $P_1$ and $P_2$ in the two motivating examples below.
	
	
	\paragraph{Example 1: Covariate shift in transfer learning.}
	In traditional machine learning, a central task is to build a predictor $\hat f:\mathcal X \mapsto \mathcal Y$ from a training sample $\{(X_{1i}, Y_{1i})\}_{i=1}^{n_1} \stackrel{i.i.d}{\sim} P_1$, and use it to predict the unseen response $Y_{2}$ given a future observation of the covariate $X_{2}$, where $(X_2,Y_2)\sim P_2$. A common assumption is that the training data distribution $P_1$ and testing data distribution $P_2$ are the same.  In many modern applications, it is often the case that the testing data may come from a different distribution, and classical methods developed for iid data need to be modified to account for such distribution difference.  Sub-fields known as domain adaptation and transfer learning have emerged to deal with scenarios in which the training data and testing data may come from different but related distributions \citep{pan2009survey,csurka2017domain,kouw2018introduction}. 
	To avoid arbitrary changes in the distribution, one situation of particular practical and theoretical interest assumes that the conditional distribution of the response given the covariate remains the same between the training and testing data while the covariates may follow different marginal distributions. This is the \emph{covariate shift} assumption. It is mathematically equivalent to the null hypothesis in \eqref{eq:problem_i}, and enables us to obtain improved predictive performance by weighting the training data according to the marginal density ratio.  This approach has been widely studied in the machine learning literature.  For examples, see \citet{shimodaira2000improving}, \citet{sugiyama2007covariate}, \citet{sugiyama2008direct}, \citet{bickel2009discriminative}, \citet{gretton2009covariate}, and \citet{sugiyama2012machine}.
	The two-sample conditional distribution testing problem is a formal way to verify the covariate shift assumption, assuming an independent sample $\{(X_{2j},Y_{2j}):1\le j\le n_2\}$ from $P_2$ is also available.
	If we do not reject the null hypothesis, the methods and theory based on the covariate shift assumption may be plausible.  Otherwise, those methods should be used with caution.
	
	\paragraph{Example 2: Causal predictive inference.}
	In the literature of causal inference, a causal predictive model is one that will work equally well under different experimental or observational environments \citep{PetersBM16,Buhlmann20,li2021searching}.  Formally, a covariate vector $X$ is called  causal for the prediction of $Y$ if $Y=h(X,\epsilon)$, for a fixed function $h$ and an independent random variable $\epsilon$ with a fixed distribution.  This definition highlights the idea that if $X$ contains all the causal factors of $Y$, then the way $X$ affects $Y$ does not depend on any other variables (potential confounders). Therefore, the conditional distribution of $Y$ given $X$ will remain the same under different experimental conditions, regardless whether these conditions are fully controlled or simply observed.  In this context, the two distributions $P_1$ and $P_2$ represent the joint distribution of $(X,Y)$ under two different experimental conditions. If we reject $H_0$ in \eqref{eq:problem_i}, then $X$ is unlikely to be causal for the prediction of $Y$.

	Most existing methods for testing conditional distributions follow one of two directions.
	The first is to test the equality of conditional moments, including semiparametric methods \citep{hardle1990semiparametric}, nonparametric methods \citep{hall1990bootstrap,kulasekera1995comparison,kulasekera1997smoothing,neumeyer2003nonparametric}, and second order moments \citep{pardo2015tests}. However, in many applications such as risk management and insurance, it is not enough to just consider mean and variance terms, and it is necessary to consider the whole conditional distribution of the response given the covariates. 
	%
	Another direction is to extend methods for unconditional distribution tests to conditional distribution tests. \citet{andrews1997conditional} extended the Kolmogorov-Smirnov test to the conditional distribution case. \citet{zheng2000consistent} proposed a test statistic based on the first-order linear expansion of Kullback-Leibler divergence. \citet{fan2006nonparametric} proposed a bootstrap test. \citet{bai2003testing} and \citet{corradi2006bootstrap} studied the problem of testing conditional distributions in a dynamic model.
	%
	%
	Most of the aforementioned methods rely on some assumptions that are hard to verify in a data-driven manner, such as smoothness of density and regression functions, and/or correctly specified parametric models. In addition, existing nonparametric methods usually involve nonparametric density estimation as an intermediate step, making them less feasible when the dimensionality is moderately high.

	Our new method for two-sample conditional distribution test has three remarkable features.  First, our test statistic is inspired from conformal prediction \citep{vovk2005algorithmic,lei2013distribution,lei2018distribution}, a framework of converting point estimators to prediction sets by exploiting the symmetry of data. The type I error control is guaranteed by a weighted exchangeability that is tailored to the null hypothesis, assuming an accurate estimate of the marginal density ratio of the covariates is available. Second, our method does not require estimating the density functions. Instead, it uses a classification algorithm to estimate the density ratio, and can incorporate almost any existing classification algorithms ranging from classical parametric estimators to modern black-box neural nets. In practice, the validity of the data-driven $p$-value depends on the accuracy of the classifier, which can be empirically validated.  This makes our method particularly useful in modern machine learning scenarios with high dimensionality and large sample sizes.  Third, the asymptotic null distribution of our test statistic and its universal power guarantee are rigorously established under certain moment conditions on the density ratios and the accuracy of classification algorithms. To our knowledge, this is the first successful extension of conformal prediction to statistical hypothesis testing beyond exchangeability with provable size and power guarantees for a data-driven procedure. 
	These theoretical results are supported by our simulation and real data examples.
	
	\paragraph{Related work in conformal prediction} We provide a general review of conformal prediction in \Cref{subsec:review}.   Roughly speaking, conformal prediction uses a conformity score to determine a sample point's agreement, or conformity, with the current data set and fitting procedure, and the resulting prediction set is the subset of the sample space with high conformity scores.  Conformal prediction in regression has been studied by \cite{lei2014distribution} in a nonparametric setting and \cite{lei2018distribution} in the high dimensional setting, where the conformity scores are chosen to be either the conditional density of $Y$ given $X$,  or the absolute fitted residual.  However, these conformity scores do not guarantee power against general alternatives.  A main methodological contribution of this paper is to show that using the conditional likelihood ratio as the conformity score can provide universal power guarantee against any alternatives.
	Such a choice of conformity score is partially inspired by a recent work of conformal prediction in classification by \cite{guan2019prediction}.  In the context of transfer learning, under the covariate shift assumption the data points are often exchangeable within the training set or the testing set alone, but not exchangeable when the training and testing datasets are merged together.  This non-exchangeability issue can be treated using the weighted conformal prediction method developed in \cite{tibshirani2019conformal}, who construct a valid $p$-value for a single observation in the test sample assuming the marginal density ratio is known.  Our method combines these ideas, with further theoretical development to allow the marginal density ratios and conformity scores to carry estimation errors.  Moreover, these existing methods only consider prediction, and our method transforms conformal prediction from a prediction tool to a hypothesis testing tool.
	
	There is a line of work on applying conformal inference in other contexts, such as testing the global null for streaming data \citep{vovk2003testing,fedorova2012plug,vovk2019testing,vovk2020testing,vovk2021retrain} and outlier detection \citep{bates2021testing}. 
	Although both the work of \citet{bates2021testing} and ours involve the conformal $p$-values and sample splitting, they are different in various aspects. \citet{bates2021testing} studied the nonparametric outlier detection problem, which is different from our goal of two-sample conditional distribution testing.  The score functions to construct the conformal $p$-values are different. Specifically, they used the one-class classification score, while we used the conditional density ratio.
	Moreover, the ways to exploit the conformal $p$-values are distinct. They constructed the $p$-values for the multiple-testing procedure, while we combined the $p$-values to eventually perform a one sided mean test. 
None of the three features of our method (testing conditional distributions, conformity score function with universal power, and asymptotic error bounds for data-driven procedures) is considered in these papers.

	\section{Background}\label{sec:prelim}
	\subsection{Problem formulation}
	If an $x\in\mathcal X$ is not in the support of $P_{1,X}$ or $P_{2,X}$, then the point $x$ should not matter in testing the conditional distribution, since the conditional distribution given this value of $x$ can be arbitrarily defined.  Therefore, to facilitate discussion we assume that $P_{1,X}$ and $P_{2,X}$ are equivalent to each other in the following sense,
	$$
	P_{1,X}\ll P_{2,X}\,,\quad \text{and}\quad P_{2,X}\ll P_{1,X}\,,
	$$
	where ``$\ll$'' stands for absolute continuity. We further assume, without loss of generality, that $P_{1,X}$, $P_{2,X}$ have density functions $f_{1,X}$, $f_{2,X}$ respectively, under a common base measure. Now we can formally state the null and alternative hypotheses as follows.
	\begin{align}\label{eq:problem}
	H_0: P_{1,X}\left\{ P_1( \cdot |X) = P_2( \cdot |X) \right\}=1, \quad \mathrm{versus} \quad H_1: P_{1,X}\left\{ P_1(\cdot|X) = P_2(\cdot|X) \right\}<1\,.
	\end{align}
	Due to the assumed equivalence between $P_{1,X}$ and $P_{2,X}$, the hypotheses in \eqref{eq:problem} can be equivalently stated by replacing $P_{1,X}$ with $P_{2,X}$.
	For a similar consideration of avoiding triviality and the ease of discussion, we also assume that $P_1(\cdot|x)$ and $P_2(\cdot|x)$ are equivalent, with density functions $f_1(y|x)$ and $f_2(y|x)$ under a common base measure.
	
	\subsection{Conformal prediction}\label{subsec:review}
	Here we briefly introduce conformal prediction in a regression setting.  For conformal prediction in other contexts, such as unsupervised learning, see \citep{vovk2005algorithmic,lei2013distribution,lei2015conformal}.
	
	Given iid data $\{(X_i, Y_i)\}_{i=1}^m$, conformal prediction converts a point estimator of the regression function $\hat \theta:\mathcal X\mapsto \mathcal Y$ to a prediction set $\hat C\in \mathcal X\times \mathcal Y$, with guaranteed finite-sample expected prediction coverage:
	\begin{align}\label{eq:conf_coverage}
	P\left\{Y_{m+1}\in\hat C(X_{m+1})\right\}\ge 1-\alpha\,,
	\end{align}
	where $\hat C(x)=\{y\in\mathcal Y: (x,y)\in\hat C\}$, and the probability is taken over the $(m+1)$-tuple of iid data $\{(X_i,Y_i):1\le i\le m+1\}$.
	
	Let $\mathcal D$ denote the sample $\{(X_i,Y_i):1\le i \le m\}$, and $\mathcal D_{-i}$ the sample obtained by removing $(X_i,Y_i)$ from $\mathcal D$.
	A conformal prediction set $\hat C$ is constructed using a conformity score function $v:(\mathcal X\times\mathcal Y)^{m+1}\mapsto \mathbb R$ that is symmetric in its first $m$ inputs. 
	For a given data set $\mathcal D$, a new $X_{m+1}$ for which a prediction of the corresponding $Y_{m+1}$ is wanted, and a $y\in\mathcal Y$, let $\mathcal D(y)$ be the augmented data set with the $(m+1)$th data point being $(X_{m+1},y)$, and $\mathcal D_{-i}(y)$ the corresponding $m$-tuple dataset with the $i$th sample removed, where the $(m+1)$th sample is $(X_{m+1},y)$. Let
	\begin{equation}\label{eq:def_V} V_i(y) = v(\mathcal D_{-i}(y)\,, ~(X_i, Y_i)), \quad i=1, \dots, m\,,~~
	V_{m+1}(y) = v(\mathcal D\,,~(X_{m+1},y))\end{equation}
	be conformity scores for each sample point in the augmented data $\mathcal D(y)$.  The conformal prediction set using the conformity score function $v(\cdot)$ is
	\begin{equation}\label{eq:conf_set}
	\hat C(X_{m+1}) =\left\{y\in\mathcal Y:\sum_{i=1}^{m+1} \mathbbm{1}\left[V_{i}(y)\le V_{m+1}(y)\right]\ge \lfloor (m+1)\alpha\rfloor \right\}\,.
	\end{equation}
	The finite sample coverage \eqref{eq:conf_coverage} can be easily derived from the iid assumption, the symmetry of $v(\cdot)$, and the construction of $\hat C$ in \eqref{eq:conf_set}.  To see this, if we replace $y$ by $Y_{m+1}$, then the iid assumption and symmetry of $v(\cdot)$ implies exchangeability of $(V_i(Y_{m+1}):1\le i\le m+1)$. Thus the rank of $V_{m+1}$ being lower than $\lfloor (m+1)\alpha \rfloor$ has probability no more than $\alpha$. 
	
	Although the finite sample coverage guarantee only requires $v(\cdot)$ to satisfy a symmetry condition, its choice will have a crucial impact on the quality of the resulting prediction set.  A good choice of $v(\cdot)$ needs to reflect the structure of the underlying distribution of $(X,Y)$ and be able to tell whether a sample point is likely drawn from this distribution.  Such a function $v(\cdot)$ is often constructed from a point estimate $\hat \theta$ of the regression function.  For example, in nonparametric regression, one can choose $v(\mathcal D, (x,y))=\hat f(y|x)$, where $\hat f(\cdot|\cdot)$ is an estimated conditional density function of $y$ given $x$ using the sample $\mathcal D\cup \{(x,y)\}$ \citep{lei2014distribution}.  In high dimensional regression, one can use $v(\mathcal D,(x,y))=-|y-\hat \theta(x)|$ where $\hat \theta$ is an estimated regression function using $\mathcal D\cup \{(x,y)\}$.
	More recently, some conformal prediction methods adaptive to heteroskedasticity based on quantile regression have been proposed \citep{romano2019conformalized,kivaranovic2020adaptive,sesia2021conformal,chernozhukov2021distributional}. In this work, we develop a new conformity score based on conditional density ratios, which is particularly suited for the two-sample conditional testing problem.
	
	
	The definition of $\hat C$ in \eqref{eq:conf_set} is only conceptual and not practical if $\mathcal Y$ is infinite, as it requires to evaluate $V_i(y)$ for all $y$ and all $1\le i\le m+1$.  For practical implementation of conformal prediction, we refer to \cite{lei2018distribution} and \cite{barber2019predictive}.  However, in our hypothesis testing problem, we do not need to actually construct a prediction set. Instead, we only need to compute the corresponding $p$-values for a subset of sample points, and evaluate their deviance from the null distribution.  The details are given in the next section as we develop our testing procedure.

	\section{A conformal test of two-sample conditional distributions}\label{sec:method}
	\subsection{The conformal $p$-value}\label{sec:conf_stat}

	Now we put the conformal prediction method described in \Cref{sec:prelim} under the context of our two-sample testing problem.
	Consider a subset of the data $\mathcal D_{(1)}=\{(X_{1i},Y_{1i}):1\le i\le n_{11}\}$ iid from $P_1$, and just a single pair $(X_{21}, Y_{21})\sim P_2$.  
	Here $n_{11}<n_1$ is a subsample size, whose value will be specified later.
	With the correspondence $m=n_{11}$, $(X_i,Y_i)=(X_{1i}, Y_{1i})$ for $1\le i\le n_{11}$ and $(X_{n_{11}+1},Y_{n_{11}+1})=(X_{21},Y_{21})$, the conformal prediction procedure described in the previous section implies that, under the simplified scenario that $P_{1,X}=P_{2,X}$, the ranking statistic
	\begin{equation}\label{eq:pi}
	\tilde U=\frac{1}{n_{11}+1}\sum_{i=1}^{n_{11}+1} \mathbbm{1}(V_i(Y_{21})\le V_{n_{11}+1}(Y_{21}))
	\end{equation}
	has an approximately $U(0,1)$ distribution under $H_0$.
	
	With a random tie-break we can make $\tilde U$ have the exact $U(0,1)$ distribution. Let $R_-=1+\sum_{i=1}^{n_{11}+1}\mathbbm 1(V_i(Y_{n_{11}+1})<V_{n_{11}+1}(Y_{n_{11}+1}))$
	and $R_+=\sum_{i=1}^{n_{11}+1}\mathbbm 1(V_i(Y_{n_{11}+1})\le V_{n_{11}+1}(Y_{n_{11}+1}))$.  Let $R$ be uniformly and independently sampled from the integers in $[R_-,R_+]$. 
	Now we can construct a uniform random variable
	\begin{equation}\label{eq:U}
	U= \frac{R-1+\zeta}{n_{11}+1}
	\end{equation}
	where $\zeta\sim U(0,1)$ is independent of everything else.  By exchangeability $R$ has a uniform distribution on $\{1,...,n_{11}+1\}$, and hence $U$ has a uniform distribution on $[0,1]$.  This $U$ can be viewed as a continuous version of $\tilde U$ in \eqref{eq:pi}, and can serve as a $p$-value for our testing problem.  Thus we call the statistic $U$ a \emph{conformal $p$-value}.
	
	
	
	In order to develop this idea into a useful test procedure, we need to resolve the following three issues.
	\begin{enumerate}
		\item How to choose a conformity score function $v$?
		\item How to allow for $P_{1,X}\neq P_{2,X}$?
		\item How to make use of multiple data points from $P_2$ to have increased power under $H_1$?
	\end{enumerate}
	
	These three issues are addressed in the next three subsections, respectively.
	
	\subsection{A choice of $v$ that separates $H_0$ and $H_1$}
	For now we still make the assumption $P_{1,X}=P_{2,X}$.
	A good choice of $v(\cdot)$ will be such that the conformal $p$-value $U$ constructed in \eqref{eq:U} has a non-uniform distribution under the alternative hypothesis $H_1$.  Common existing choices such as conditional density and absolute residual do not satisfy this.  Our choice of $v(\cdot)$ is the conditional Radon-Nikodym derivative between $P_1$ and $P_2$:
	$$
	v(x,y) = \frac{dP_{1}(y|x)}{dP_2(y|x)}=\frac{f_1(y|x)}{f_2(y|x)}\,.
	$$
	This $v$ function is different from the conformity score functions introduced in \eqref{eq:def_V}, as it only involves a single data pair $(x,y)$. This is not a real problem as we can let $v$ be independent of the first $m$ arguments and treat the input $(x,y)$ as the last argument. 
	
	\begin{remark}\label{rem:est_V}
		The function $v$ involves the unknown density ratio $f_1(y|x)/f_2(y|x)$.  Our method will need to use an empirical version of $v$:
		$$
		\hat v(x,y)=\widehat{\frac{f_1(\cdot|\cdot)}{f_2(\cdot|\cdot)}}(x,y)\,,
		$$
		where $\widehat{\frac{f_1(\cdot|\cdot)}{f_2(\cdot|\cdot)}}$ is an  estimate of the conditional density ratio, independent of $\{(X_{1i},Y_{1i}):1\le i\le n_{11}\}$ and $(X_{21},Y_{21})$.  A remarkable advantage of our choice of $v$ and $\hat v$ is that the density ratio $f_1(\cdot|\cdot)/f_2(\cdot|\cdot)$ can be conveniently estimated using classification algorithms, which is both theoretically and practically much easier than estimating the density functions themselves. 
		There is a rich literature on classification and density ratio estimation, with many powerful algorithms even in high dimensional settings. Further discussion of estimating the conditional density ratio is provided in \Cref{sec:alg} when we summarize our algorithm.
	\end{remark}
	
	The ability of $v(x,y)$ to separate $H_0$ and $H_1$ is established by the following lemma.
	\begin{lemma}[Separation of $H_0$ and $H_1$ by $v$ under equal $X$-marginal]\label{lem:separation}
		If $P_{1,X}=P_{2,X}$ and $v(x,y)=\frac{f_1(y|x)}{f_2(y|x)}$,
		then under $H_1$, $\mathbb E U = \frac{1}{2}-\frac{1}{4}\mathbb E|v(X_2,Y_2)-v(X_2',Y_2')|<\frac{1}{2}$ for all values of $n_{11}\ge 1$, where $(X_2,Y_2)$, $(X_2',Y_2')$ are iid realizations from $P_2$.
	\end{lemma}
	\Cref{lem:separation} can be viewed as a special case of \Cref{lem:weighted_conf}(c), for which a complete proof is provided in \Cref{sec:proof_main}.  Under $H_0$, $v(x,y)\equiv 1$.  Under $H_1$, the conditional likelihood ratio $v(X_2,Y_2)$ cannot be degenerate, so $\mathbb E|v(X_2,Y_2)-v(X_2',Y_2')|$ is strictly positive and hence we have $\mathbb E U<1/2$ under $H_1$. This also suggests a simple one sided rejection rule for our test.
	
	\begin{remark}\label{rem:choice_V}
		The choice of $v$ is motivated from an information theoretic perspective.
		It can be directly verified that 
		$\mathbb E_{P_1} v(X,Y)-\mathbb E_{P_2}v(X,Y)=\mathbb E_{P_1} \left[D_{\chi^2}(f_1(\cdot|X),f_2(\cdot|X))\right]$, where
		$D_{\chi^2}(f_1,f_2)=\int f_1^2/f_2-1$ is the Neyman's $\chi^2$ divergence between two densities $f_1,f_2$.
		As a result, $\mathbb E_{P_1} v(X,Y)-\mathbb E_{P_2}v(X,Y)\ge 0$ with equality holds if and only if $H_0$ is true. This suggests that under the alternative, $v(X,Y)$ tends to take larger values under $P_1$ than under $P_2$.
		A more involved argument is needed in order to carry over this intuition rigorously to analyze the rank of $V_{n_{11}+1}(Y_{21})$ in \eqref{eq:pi} and the continuous version in \eqref{eq:U}.  It is made clear in \eqref{eq:TV_link}, \Cref{sec:local_alt}, that the conformal $p$-value based on $v(x,y)=f_1(y|x)/f_2(y|x)$ is closely related to the expected conditional total variation distance between $P_1$ and $P_2$. 
	\end{remark}

	\subsection{Allowing for $P_{1,X}\neq P_{2,X}$ using weighted conformalization}
	Now we drop the assumption of equal marginal distribution of $X$ under $P_1$ and $P_2$. Recall the notation  $(X_{i},Y_{i})_{i=1}^{n_{11}+1}$, with $(X_i,Y_i)\sim P_1$ for $1\le i\le n_{11}$ and $(X_{n_{11}+1},Y_{n_{11}+1})\sim P_2$.  Now the $(n_{11}+1)$-tuple used to construct $U$ in \eqref{eq:U} are no longer exchangeable under the null hypothesis and the distribution of $U$ will in general not be uniform.  Here we use the ``weighted conformal prediction'' idea developed in \cite{tibshirani2019conformal} to obtain a modified version of $U$ with valid uniform sampling distribution under $H_0$. The key idea is to condition on a randomly permuted data sequence.
	
	In the subsequent discussion, we will focus on the conformity score functions $v$ that only depends on the last argument, and write $V_i=v(X_i,Y_i)$.
	
	We begin by imagining that the data $\mathbf Z=(X_i,Y_i)_{i=1}^{n_{11}+1}$ are stored in two parts: a randomly permuted sequence $\tilde{\mathbf Z}=(\tilde X_i,\tilde Y_i)_{i=1}^{n_{11}+1}$, and the permutation $\sigma:[n_{11}+1]\mapsto [n_{11}+1]$ with the correspondence $(X_i,Y_i)\leftrightarrow (\tilde X_{\sigma(i)},\tilde Y_{\sigma(i)})$.
	By construction, the vector $(V_i:1\le i\le n_{11}+1)$ is a deterministic function of $(\tilde{\mathbf Z},\sigma)$.  
	Given $\tilde{\mathbf Z}$, $V_{n_{11}+1}$ may take $n_{11}+1$ possible values, and $V_{n_{11}+1}=v(\tilde X_i,\tilde Y_i)$ if $\sigma(n_{11}+1)=i$.
	
	Now we are ready to derive the conditional distribution of $V_{n_{11}+1}$ given $\tilde{\mathbf Z}$, construct the uniformly distributed weighted conformal $p$-value, and establish its ability to separate $H_0$ and $H_1$.
	\begin{lemma}\label{lem:weighted_conf}
		\begin{enumerate}
			\item [(a)] Under $H_0$, for any choice of $v(x,y)$ we have $$(V_{n_{11}+1}|\tilde{\mathbf Z})\sim\sum_{i=1}^{n_{11}+1}p_i(\mathbf Z)\delta_{V_i}$$ with
			\begin{equation*}
			p_i(\mathbf Z) =
			\frac{\frac{f_{2,X}(X_{i})}{f_{1,X}(X_{i})}}{\sum_{l=1}^{n_{11}+1} \frac{f_{2,X}(X_{l})}{f_{1,X}(X_{l})}}\,,\quad i=1,...,n_{11}+1\,, \end{equation*}
			where $f_{k,X}(\cdot)$ denotes the marginal density function of $X$ under $P_k$ ($k=1,2$), and $\delta_{v}$ denotes the point mass at $v$.
			\item [(b)]  For any choice of $v(x,y)$, the randomized statistic
			\begin{equation}\label{eq:U_weighted}
			U = \sum_{i=1}^{n_{11}+1} p_i(\mathbf Z)\mathbbm{1}(V_i< V_{n_{11}+1}) + \zeta \sum_{i=1}^{n_{11}+1} p_i(\mathbf Z)\mathbbm{1}(V_i=V_{n_{11}+1})
			\end{equation}
			has a uniform distribution under $H_0$, where $\zeta$ is a $U(0,1)$ random variable independent of everything else.
			\item [(c)] Under $H_1$, if $v(x,y)=f_1(y|x)/f_2(y|x)$, there exist $\delta >0$ and $m_0>0$, depending only on $P_1$ and $P_2$, such that $\mathbb E U \le 1/2 - \delta$ when $n_{11} \ge m_0$.
		\end{enumerate}
	\end{lemma}
	The definitions of $U$ in \eqref{eq:U_weighted} and \eqref{eq:U} are compatible, as the construction with random tie-breaking in \eqref{eq:U} can be viewed as a special case of \eqref{eq:U_weighted} with $p_i(\mathbf Z)=(n_{11}+1)^{-1}$.
	
Part (a) of \Cref{lem:weighted_conf} is due to \cite{tibshirani2019conformal}, who first considered weighted conformal prediction under the covariate shift assumption.
	Part (b) is a simple consequence of part (a), which can be viewed as a discrete version of the CDF transform.  The most non-obvious part of the proof is that of part (c), which exploits the form of $v(x,y)=f_1(y|x)/f_2(y|x)$. The detailed proofs are given in \Cref{sec:proof_main}. 

	The conformal $p$-value $U$ will exhibit a constant difference from the null distribution once the ranking sample size $n_{11}$ exceeds a finite threshold.  Such a separation between the null and alternative hypotheses can be amplified using multiple such weighted conformal $p$-values, as we discuss in the next subsection. 
		
	\subsection{Incorporating multiple testing sample points for better power}
	So far we have only focused on obtaining a single conformal $p$-value from a single sample point in $P_2$.  Such a single $p$-value often has limited power in distinguishing $H_1$ from $H_0$.  To have a consistent test that rejects $H_0$ under the alternative hypothesis with probability tending to $1$ as the sample size increases to $\infty$, we must consider multiple testing sample points: $\{(X_{2j},Y_{2j}):1\le j\le n_{21}\}$, where $n_{21}$ is a subsample size whose relationship with the original sample size $n_2$ will be discussed later.
	
	Now assume that we have obtained estimates $\hat g$ and $\hat v$ for the marginal and conditional density ratios $g(x)\equiv f_{2,X}(x)/f_{1,X}(x)$ and $v(x,y)=f_1(y|x)/f_2(y|x)$, respectively.
	Given $(X_{1i},Y_{1i})_{i=1}^{n_{11}}$ from $P_1$ and $(X_{2j},Y_{2j})_{j=1}^{n_{21}}$ from $P_2$, we can repeat the procedure used to obtain $U$ in \eqref{eq:U} for each sample point $(X_{2j},Y_{2j})$ for $1\le j\le n_{21}$, resulting in $(\hat U_{j}:1\le j\le n_{21})$. If the function estimates $\hat g$ and $\hat v$ are accurate enough, then approximately each $\mathbb E \hat U_j = 1/2$  under $H_0$ and  $\mathbb E \hat U_j < 1/2$ under $H_1$.  However, these $\hat U_j$'s are dependent as they use the same set of ranking sample $(X_{1i},Y_{1i})_{i=1}^{n_{11}}$ from the first population. To obtain a valid $p$-value for one-sided mean test over the $\hat U_j$'s, we must take their dependence into account.  
	To this end, we re-define $\hat U_j$ as
	$$ \hat U_j = \frac{\frac{1}{n_{11}}\sum_{i=1}^{n_{11}} \hat g(X_{1i})\hat D_{ij}}{\frac{1}{n_{11}}\sum_{i=1}^{n_{11}} \hat g(X_{1i})}, ~~ j=1,\dots, n_{21} \,, $$
	where $\hat D_{ij}=\{\mathbbm{1}(\hat V_{1i} < \hat V_{2j}) + \zeta_j \mathbbm{1}(\hat V_{1i}=\hat V_{2j})\}$, $\hat V_{1i}=\hat v(X_{1i}, Y_{1i})$ and $\hat V_{2j} = \hat v(X_{2j}, Y_{2j})$.   Such a $\hat U_j$ can be viewed as an approximate plug-in version of the conformal $p$-value in \eqref{eq:U_weighted}, with estimated versions of $g$ and $v$, and omitting the vanishing terms $\hat g(X_{2j})/n_{11}$.
	
	The key observation is that despite the dependence due to a common ranking sample, the average of these $p$-values
	$n_{21}^{-1}\sum_{j=1}^{n_{21}}\hat U_j$ is a two-sample U-statistic conditioning on the estimated density ratios $\hat g$, $\hat v$, whose asymptotic distribution can be readily estimated using plug-in estimators.
	
	Formally, we use statistic
	\begin{equation}\label{eq:hat_T}
	\hat T=\frac{\frac{1}{2}-\frac{1}{n_{21}}\sum_{j=1}^{n_{21}}\hat U_j}{\hat\sigma/\sqrt{n_{11}}}
	\end{equation}
	where $\hat\sigma$ is the estimated asymptotic standard deviation of $\sqrt{n_{11}}n_{21}^{-1}\sum_{j=1}^{n_{21}}\hat U_j$.
	
	Let $\hat F_n$ be the empirical CDF of $\{\hat V_{2j}:1\le j\le n_{21}\}$, and $\hat F_{n,1/2}=(\hat F_n+\hat F_{n,-})/2$ where $F_-$ is the left limit of a function $F$.
	The asymptotic variance used in $\hat T$ can be estimated as follows.
	\begin{align}\label{eq:hat_sigma}
	\hat\sigma^2 = \hat\sigma_1^2+\frac{n_{11}}{12n_{21}}+\frac{1}{4}\hat\sigma_2^2-\hat\rho_{12}\,,  
	\end{align}
	where
	$\hat\sigma_1^2$ is the empirical variance of $\{\hat g(X_{1i})[1-\hat F_{n,1/2}(\hat V_{1i})]:1\le i\le n_{11}\}$,
	$\hat\sigma_2^2$ is the empirical variance of $\{\hat g(X_{1i}):1\le i\le n_{11}\}$,
	and $\hat\rho_{12}$ is the empirical covariance between $\{\hat g(X_{1i})[1-\hat F_{n,1/2}(\hat V_{1i})]:1\le i\le n_{11}\}$ and $\{\hat g(X_{1i}):1\le i\le n_{11}\}$. The derivation of the asymptotic variance is provided in \Cref{thm:null-shift} and \Cref{lem:hat_U_U_prime}.  
	
	\begin{remark}
		The asymptotic variance of $\hat T$ can be alternatively estimated using $\hat g(X_{2j})$ and $[1-\hat F_{n,1/2}(\hat V_{2j})]$ using an importance sampling technique.  This provides a larger sample size for asymptotic variance estimation.  Practically, we found using the harmonic mean of the original estimate and the importance sampling estimate to have good performance in simulations.  The details of the importance sampling estimate are given in \Cref{sec:asy_var_aux}.
	\end{remark}
	
	\subsection{The conformal conditional distribution test algorithm}\label{sec:alg}
	Given the ideas and methods presented in the previous subsections, we can now 
	describe the full testing procedure in \Cref{alg:shift} below.
	The algorithm assumes availability of two classification subroutines: (i) a marginal classification algorithm $\mathcal A_1$ that takes input two labeled samples $\{X_{1i}:n_{11}+1 \le i\le n_{1}\}$, $\{X_{2j}:n_{21}+1 \le j\le n_{2}\}$ with sample sizes $n_{12}=n_1-n_{11}$, $n_{22}=n_2-n_{21}$ respectively, and outputs an estimate of the marginal density ratio $g=f_{2,X}/f_{1,X}$;
	and (ii) a joint classification algorithm $\mathcal A_2$ that takes input two labeled samples $\{(X_{1i},Y_{1i}):n_{11}+1\le i\le n_{1}\}$ and $\{(X_{2j},Y_{2j}):n_{21}+1 \le j\le n_{2}\}$, and outputs an estimate of the conditional density ratio $v=f_1(y|x)/f_2(y|x)$.  Our numerical experiments use a equal split ratio: $(n_{11},n_{21})=(n_1/2, n_2/2)$, which yields reasonable performance in all the scenarios considered.
	
	Given $\mathcal A_1$, $\mathcal A_2$, the testing procedure first splits the sample, applying the density ratio estimation subroutines $\mathcal A_1$, $\mathcal A_2$ on one part to obtain approximate versions of the density ratios.  Then the other part is used to obtain the final test statistic $\hat T$.

	\begin{algorithm}
		\caption{\label{alg:shift} Two-sample test of conditional distribution}
		\begin{algorithmic}
			\REQUIRE Training data $(X_{1i}, Y_{1i})_{i=1}^{n_1}$; testing data $(X_{2j}, Y_{2j})_{j=1}^{n_2}$; density ratio estimation subroutines $\alg_1$, $\alg_2$
			\STATE For $k=1,2$, randomly split $\{1, \dots, n_k\}$ into subsets $\sdata_{k1}=\{1, \dots, n_{k1}\},~ \sdata_{k2}=\{n_{k1}+1, \dots, n_k\}$
			\STATE $\hat{g}(\cdot) = \alg_1[\{X_{1i}, i \in \sdata_{12}, X_{2j}, j \in \sdata_{22}\}]$
			\STATE $\hat{v}(\cdot,\cdot) = \alg_2[\{(X_{1i}, Y_{1i}), i \in \sdata_{12}, (X_{2j}, Y_{2j}), j \in \sdata_{22}\}]$
			\FOR{$j\in\mathcal I_{21}$}
			\STATE Generate $\zeta_{j} \sim U(0,1)$, independent of everything else
			\STATE $\hat{D}_{ij} = \mathbbm{1}(\hat{V}_{1i} < \hat{V}_{2j})+\zeta_j \mathbbm{1}(\hat{V}_{1i}=\hat{V}_{2j})$, where $\hat V_{ki}=\hat v(X_{ki},Y_{ki})$
			\STATE $\hat{U}_j = \sum_{i=1}^{n_{11}}\hat{g}(X_{1i}) \hat{D}_{ij} /\sum_{i=1}^{n_{11}}\hat{g}(X_{1i})$
			\ENDFOR
			\STATE $\hat T =(1/2 -n_{21}^{-1}\sum_{j=1}^{n_{21}}\hat{U}_j)/(\hat{\sigma}/\sqrt{n_{11}})$, where $\hat \sigma^2$ is given in \eqref{eq:hat_sigma} 
			\STATE Reject $H_0$ if $\hat T \geq \Phi^{-1}(1-\alpha)$ ($\Phi$ is the CDF of $N(0,1)$, $\alpha$ is the nominal type I error level)
		\end{algorithmic}
	\end{algorithm}
	
	
	\paragraph{Simplification when the marginals are equal.}  Sometimes it is plausible to assume that the marginal distributions $f_{1,X}$, $f_{2,X}$ are equal.  This can happen, for example, when the sampling schemes and environments of $X$ are known to be the same, or they come from the same experimental design.  In this scenario, the algorithm becomes much simpler, as we know that $f_{1,X}/f_{2,X}\equiv 1$. As a result, the algorithm does not need to use the marginal density ratio subroutine $\mathcal A_1$ and $g \equiv 1$.
	
	
	\paragraph{Choice of classification algorithms.} As mentioned in \Cref{rem:est_V}, our method does not require estimating the densities $f_{k,X}(\cdot)$ or conditional densities $f_{k}(\cdot|\cdot)$ for $k=1,2$.  Instead, it only requires estimating the marginal density ratio $f_{1,X}(x)/f_{2,X}(x)$, and the conditional density ratio $f_{1}(y|x)/f_2(y|x)$ only needs to be estimated up to a monotone transform, since only the ranking information is needed in the test statistic.  Estimating density ratios is often easier than estimating the density functions themselves, and has been well studied in the statistics and machine learning literature, including moment matching approach \citep{gretton2009covariate}, the ratio matching approach \citep{sugiyama2008direct,kanamori2009least,tsuboi2009direct}, and probabilistic classification approach \citep{qin1998inferences,cheng2004semiparametric,bickel2007discriminative}. 
	
	Our algorithm can be implemented with any available density ratio estimators.
	Here we provide some further detail about the probabilistic classification estimator due to its simplicity.   In the case of $f_{1,X}=f_{2,X}$, we only need to consider a single classification problem over the joint distribution $(X,Y)$, where class ``$1$'' represents the subsample $\{(X_{1i},Y_{1i}):i\in \mathcal I_{12}\}$, and class ``$2$'' represents the subsample $\{(X_{2j},Y_{2j}):j\in \mathcal I_{22}\}$.   Let $\eta(x,y)$ be the true conditional probability $P(1|x,y)$, then $\eta(x,y)/(1-\eta(x,y))\propto f_1(x,y)/f_2(x,y)$, which also equals $f_{1}(y|x)/f_2(y|x)$ since $f_{1,X}=f_{2,X}$.   When $f_{1,X}\neq f_{2,X}$, we can consider an additional classification problem using only $X$. Let $\eta(x)=P(1|x)$, then $f_{2,X}(x)/f_{1,X}(x)= n_{12}(1-\eta(x))/(n_{22}\eta(x))$.  With probabilistic classifiers providing $\hat\eta(x,y)$ and $\hat\eta(x)$, the corresponding joint and marginal density ratios can be estimated by plugging in $\hat\eta(x,y)$ and $\hat\eta(x)$.  The conditional density ratio can be obtained by taking a further ratio between the joint and marginal density ratios.
	Many commonly used classification methods offer a probability output, including the classical linear and quadratic discriminant analysis, logistic regression, popular machine learning algorithms such as random forest and support vector machines \citep{sollich2000probabilistic}, and modern deep neural nets.
	
	\section{Asymptotic Properties}\label{sec:thm}
	
	In this section, we investigate the theoretical properties of the testing procedure described in \Cref{alg:shift} under (i) a standard fixed population asymptotic framework, and (ii) a local alternative perspective.
	Since the test statistic is constructed from the ranking subsamples, we assume that the two ranking subsample sizes are proportional: $n_{11}/n_{21}$ stays bounded and bounded away from $0$ as $n_{11}\rightarrow\infty$.
	The asymptotic behavior of our test will depend on the estimated functions $\hat g$ and $\hat v$, which depends on the fitting sample sizes $(n_{12},n_{22})$.  It is natural to have $(n_{12}, n_{22})$ increasing with the ranking sample size $(n_{11},n_{21})$.  We quantify the required accuracy of the density ratio estimates $\hat g$, $\hat v$ in \Cref{ass:accuracy_h_g} below.
	Designing a high quality density ratio estimator is a rich and context-dependent topic, and is beyond the scope of this paper.
	
	\subsection{Fixed population asymptotics}
	We first consider the classical setting where the two distributions $P_1$, $P_2$ are fixed, and study the limiting behavior of the test statistic as the ranking sample size $n_{11}$  grows to infinity and $n_{11}/n_{21}$ stays bounded and bounded away from $0$.  A local alternative analysis with varying signal strength is presented in \Cref{sec:local_alt} below.
	
	Recall that we use the notation $g(x)=f_{2,X}(x)/f_{1,X}(x)$ and $v(x,y)=f_1(y|x)/f_2(y|x)$.  
	For $k=1,2$, let $G_{ki}=g(X_{ki})$ and $V_{ki} = v(X_{ki}, Y_{ki})$. Let $D_{ij}=\indc(V_{1i}<V_{2j})+\zeta_j\indc(V_{1i}=V_{2j})$ where $\zeta_j$'s are auxiliary $U(0,1)$ random variables independent of everything else. Define $\hat G_{ki}$, $\hat V_{ki}$ and $\hat D_{ij}$ similarly as $G_{ki}$, $V_{ki}$ and $D_{ij}$ using the estimated functions $\hat g$, $\hat v$.  For a random variable $Z$ and constant $q>0$, $\|Z\|_q$ denotes the $\ell_q$ norm of $Z$: $\|Z\|_q^q=\mathbb E(|Z|^q)$.  Much of our analysis will involve the estimation errors in $\hat g$, $\hat v$ reflected through the random variables $\hat G_{11}-G_{11}$ and $\hat D_{11}-D_{11}$.  We use $\mathbb E_*(\cdot)$, ${\rm Var}_*(\cdot)$, and $\|\cdot\|_{q,*}$ to denote the conditional expectation, variance, and $\ell_q$ norm given the density ratio estimates $\hat v$, $\hat g$ (or, equivalently, given the fitting subsample). For example, $\mathbb E_*(\hat v(X_1, Y_1)) = \mathbb E (\hat v(X_1, Y_1)|\hat v)$, and $\|\hat G_{11}-G_{11}\|_{2,*}=[\mathbb E(\hat g(X_{11})-g(X_{11}))^2|\hat g]^{1/2}$.
	
	Our first assumption puts some moment conditions on the marginal density ratio $g(\cdot)$.
	\begin{assumption}\label{ass:moment}
		The marginal likelihood ratio $g(x)=f_{2,X}(x)/f_{1,X}(x)$ satisfies
		$\|G_{11}\|_2<\infty$.
	\end{assumption}
	
	Our next assumption is on the asymptotic accuracy of the density ratio estimators.
	\begin{assumption}\label{ass:accuracy_h_g}
		\begin{enumerate}
			\item  [(a)] $\|\hat G_{11}-G_{11}\|_{2,*}=o_P(1)$.
			\item [(b)] $\left|\mathbb E_*(\hat G_{11}-G_{11})\hat D_{11} - \mathbb E_*(\hat G_{11}-G_{11}) \mathbb{E}_*(G_{11}\hat D_{11})\right| = o_P(1/\sqrt{n_{11}}).$
		\end{enumerate}
	\end{assumption}
To help understand the notation, consider \Cref{ass:accuracy_h_g}(a) for example.  By construction and the definition of $\|\cdot\|_{2,*}$, we have
$\|\hat G_{11}-G_{11}\|_{2,*}=[\mathbb E(\hat g(X_{11})-g(X_{11}))^2|\hat g]^{1/2}$, which is a random quantity, whose randomness comes from $\hat g$, which is itself a function of the fitting subsamples. Therefore, \Cref{ass:accuracy_h_g}(a) requires that with high probability over the randomness of the fitting subsamples, the estimate $\hat g$ is close to $g$ when their distance is measured using the $\ell_2$ norm of $\hat g-g$ under $f_{1,X}$.  Notation in part (b) of \Cref{ass:accuracy_h_g} is interpreted accordingly.

    \Cref{ass:accuracy_h_g}(a) requires consistent estimation of the marginal density ratio $f_{1,X}(x)/f_{2,X}(x)$, which is mild.  \Cref{ass:accuracy_h_g}(b) deserves further discussion, which is deferred after the presentation of the main theorem on the asymptotic behavior of the test statistic output by \Cref{alg:shift}.
	\begin{theorem}\label{thm:null-shift}
		Suppose that Assumptions \ref{ass:moment} and \ref{ass:accuracy_h_g} hold. The test statistic $\hat T$ output by \Cref{alg:shift} converges in distribution to the standard normal as $n_{11}\to \infty$ under $H_0$.	
		
		Let $\Delta=-\mathbb E_*(\hat G_{11}-G_{11})\hat D_{11} + \mathbb E_*(\hat G_{11}-G_{11}) \mathbb{E}_*(G_{11}\hat D_{11})+\mathbb E_* G_{11}(\hat D_{11}-D_{11})$.
		If there exists a constant $c>0$ such that
		$$\mathbb P\left[\Delta <(1/4)\mathbb E|v(X_2,Y_2)-v(X_2',Y_2')|-c\right]\rightarrow 1\,,$$ where $(X_2,Y_2)$ and $(X_2',Y_2')$ are iid realizations from $P_2$, then under $H_1$, $\hat T\rightarrow\infty$ in probability.
	\end{theorem}
	
	
Now we discuss \Cref{ass:accuracy_h_g}(b), which is needed to ensure that the approximation error in the estimated weights does not break the central limit theorem.  Consider the following two scenarios.
    \begin{enumerate}
    	\item $\sqrt{n_{11}}\|\hat G_{11}-G_{11}\|_{2,*}=o_P(1)$.  In this case \Cref{ass:accuracy_h_g}(b) follows immediately from boundedness of $\hat D$ and Cauchy-Schwartz, regardless of $\hat v$, the estimated conditional density ratio.  This reflects the typical validity guarantee for conformal methods: when the weights are accurate enough, the type I error is always controlled for all conformity score functions.  However, in order to achieve such a convergence rate of $\|\hat G_{11}-G_{11}\|_{2,*}$, we typically will need the fitting subsample size $n_{12}$ to be much larger than the ranking subsample size $n_{11}$.
    	In the special case where we have side information $f_{1,X}=f_{2,X}$, then there is no need to estimate $\hat g$, and $\hat g\equiv1$ and  \Cref{ass:accuracy_h_g} holds trivially. 
    	\item $\sqrt{n_{11}}\|\hat G_{11}-G_{11}\|_{2,*}$ does not converge to $0$ in probability. This is more likely to be the case when $n_{12}\asymp n_{11}$, and is of major practical interest.  In this case, the convergence of $\hat G_{11}-G_{11}$ alone is not enough to control the approximation error in the weights.  In order for \Cref{ass:accuracy_h_g}(b) to hold we will need the random variable $\hat D_{11}$, which involves the estimated conditional density ratio $\hat v$, to behave reasonably.  Specifically, after some simple algebra, the left hand side of \Cref{ass:accuracy_h_g}(b) equals $$|\mathrm{Cov}_*(\hat G_{11}-G_{11}, \hat D_{11})-\mathbb E_*(\hat G_{11}-G_{11})\mathrm{Cov}_*(G_{11},\hat D_{11})|\,,$$
		which is upper bounded by (ignoring constant factors) $\rho\|\hat G_{11}-G_{11}\|_{2,*}$, with
    	$$
        \rho = \max\left(|{\rm Corr}_*(\hat G_{11}-G_{11}, \hat D_{11})|,~|{\rm Corr}_*(G_{11}, \hat D_{11})|\right)\,.
    	$$ There is good reason to expect ${\rm Corr}_*(\hat G_{11}-G_{11}, \hat D_{11})$ and ${\rm Corr}_*(G_{11}, \hat D_{11})$ to be close to $0$, because for a good estimate $\hat v$, the randomness in $\hat v(X_{11},Y_{11})$ will be mostly from the conditional randomness of $Y_{11}$ given $X_{11}$, which is independent of $X_{11}$ itself. But both $\hat G_{11}$ and $G_{11}$ are functions of $X_{11}$, so we should expect $\hat v(X_{11},Y_{11})$ (and hence $\hat D_{11}$) to be nearly independent of $\hat G_{11}$ and $G_{11}$.  Consider a simple Gaussian mean shift example, where $X_{11}\sim N(-\delta/2,1)$, $(Y_{11}|X_{11})\sim N(X_{11}-\mu/2, 1)$, $X_{21}\sim N(\delta/2, 1)$, $(Y_{21}|X_{21})\sim N(X_{21}+\mu/2,1)$. Suppose all functions are estimated using the parametric maximum likelihood, then $\hat D_{11}=\mathds{1}[\hat\mu(Y_{11}-X_{11})>\hat\mu(Y_{21}-X_{21})]$, which is independent of $X_{11}$, because, by construction and joint Gaussianity, $Y_{11}-X_{11}$ is independent of $X_{11}$.  This example can be extended to more complex versions, such as heteroskedastic responses.  We expect that a weak dependence between $\hat D_{11}$ and $X_{11}$ also holds true for other good conditional density ratio estimators.  In our numerical experiments, we observe that such weak dependence assumption is indeed plausible in many settings. See \Cref{sec:error_quantities} for detailed empirical evidences.  In practice, if we are confident about the density ratio estimate, such as when using a correctly specified parametric model, then an equal-sized sample split for fitting and ranking is recommended.  Otherwise, one could use a larger sample size for fitting and a smaller sample size for ranking.  We report empirical results for different fitting-ranking split ratios in \Cref{sec:other_split_ratio}.
    \end{enumerate}

	
	The asymptotic power guarantee  only requires $\hat D_{11}$ and $\hat G_{11}$ to be within a constant distance from their corresponding population versions.  This is because when there is a constant separation between the data distribution and the null model, a small constant distortion in the test statistic will not remove all the signal.  Below we provide a more delicate analysis under a local alternative framework.
	
	\subsection{A local alternative analysis}\label{sec:local_alt}
	Now we consider a local alternative scenario such that the two joint distributions $(P_1,P_2)$ may change with the sample size $(n_{11},n_{21})$.  This provides a more refined view of different sources of the approximation error and the accumulation of signal strength when the sample size increases.  
	
	Following \Cref{thm:null-shift} and \Cref{lem:weighted_conf}(c),
	we use the following quantity to quantify the deviation from the null, 
	$$
	\delta=\delta(P_1,P_2)=\mathbb E|v(X_2,Y_2)-v(X_2',Y_2')|\,,
	$$
	where $(X_2,Y_2),(X_2',Y_2')\stackrel{iid}{\sim} P_2$.  By construction $\mathbb E v(X_2,Y_2)=1$ and $\delta=0$ if and only if $\mathbb P( v(X_2,Y_2)=1)=1$, which is equivalent to $H_0$.  A larger value of $\delta$ indicates the conditional density ratio $f_1(Y|X)/f_2(Y|X)$ is likely to be far away from $1$.  In fact, using the triangle inequality and Jensen's inequality we can show that
	\begin{equation}\label{eq:TV_link}
	\mathbb E_{X\sim f_{2,X}} D_{\rm tv}(f_1(\cdot|X),f_2(\cdot|X))=\frac{1}{2}\mathbb E|v(X_2,Y_2)-1|\in [\delta/4,\delta/2]\,,
	\end{equation}
	where $D_{\rm tv}$ is the total variation distance between two distributions.  A detailed proof of this claim is given in \Cref{sec:proof_main}.
	
	Under a slightly stronger technical condition than those in \Cref{thm:null-shift}, we have the following local alternative result.
	\begin{proposition}\label{pro:local_alt}
		In addition to Assumptions \ref{ass:moment} and \ref{ass:accuracy_h_g}, suppose that $\left|\mathbb E_* G_{11}(\hat D_{11}-D_{11})\right|=o_P(n_{11}^{-1/2})$, $\hat v(X_{21},Y_{21})-v(X_{21},Y_{21})=o_P(1)$, and $v(X_{21},Y_{21})$ has a continuous distribution with bounded density.
		Then, we have
		\begin{align*}
		\hat T = \frac{\sqrt{n_{11}}\delta}{4\sigma}(1+o_P(1))+Z+o_P(1)\,,
		\end{align*}
		where $\sigma$ is the population version of $\hat \sigma$ in \eqref{eq:hat_sigma} using the true density ratios $g$ and $v$, and $Z\rightsquigarrow N(0,1)$ as $n_{11}\rightarrow\infty$. 
	\end{proposition}{}
	The most non-trivial additional assumption here is $\left|\mathbb E_* G_{11}(\hat D_{11}-D_{11})\right|=o_P(n_{11}^{-1/2})$, whereas only a constant error bound is required in \Cref{thm:null-shift} for the test to be powerful against a constant alternative.  This is because the local alternative is very close to the null, while in the fixed population analysis considered in \Cref{thm:null-shift}, the difference between the null and alternative is much larger.  This more stringent assumption can still be realistic for the same reason explained in the discussion after \Cref{thm:null-shift}. Numerical evidences are provided in \Cref{sec:error_quantities}.
	The continuity of $v(X_2,Y_2)$ and bounded density allow us to provide more refined control on the difference between indicators $\hat D_{11}-D_{11}$.

	\Cref{pro:local_alt} suggests the following  local asymptotic behavior of our test statistic.
	\begin{enumerate}
		\item If $\delta\sqrt{n_{11}}\rightarrow\infty$, then $\hat T\rightarrow\infty$ in probability.
		\item If $\delta\sqrt{n_{11}}\rightarrow a\in[0,\infty)$, then $\hat T\rightsquigarrow N(a/(4\sigma), 1)$.
	\end{enumerate}
	
	As a result, in the asymptotic regime considered here, the power is mostly determined by $\delta\sqrt{n_{11}}$.
	
	\section{Simulation study}\label{sec:sim}
	
	In this section, we illustrate the performance of our method in several simulation settings. 
	For brevity, we focus on the more challenging and interesting case where the $X$-marginals are different.
	Denote $x_{i}=(x_{i}(1), \dots, x_{i}(p))^{\t}$, and $I_p$ is a $p \times p$ identity matrix. 
	We first consider three prototypical regression models with $p=5$ that are similar to those in \citet{lei2018distribution} and \citet{zheng2000consistent}. A higher dimensional case is presented in \Cref{sec:sim_hi}.
	
	{\bf Model A} (Gaussian, linear). Let $y_{\ell i}=\alpha_{\ell} + \beta^{\t}x_{\ell i} + \epsilon_{\ell i}$, $i = 1, \dots, n_{\ell}, \ell=1,2$, where $x_{1i} \stackrel{iid}{\sim} N(\mathbf{0}, I_p)$, $x_{2i} \stackrel{iid}{\sim} N(\mu, I_p)$ where $\mu=(1,1,-1,-1,0)^{\t}$, and $\epsilon_{1i}, \epsilon_{2i} \stackrel{iid}{\sim} N(0, 1)$, independent of the features. Set $\alpha_1=\alpha_2=0$ under the null and $\alpha_1=0, \alpha_2=0.5$ under the alternative.
	\vspace{0.05in}
	
	{\bf Model B} (Gaussian mixture, nonlinear, heavy-tailed). Let $y_{\ell i}= \alpha_{\ell} + \beta_{1}x_{\ell i}(1) + \beta_{2}x_{\ell i}(2) + \beta_{3}x_{\ell i}^2(3) + \beta_{4}x_{\ell i}^2(4) + \beta_{5}x_{\ell i}^3(5) + \epsilon_{\ell i}, i=1,\dots,n_{\ell}, \ell=1,2$, where $x_{1i} \stackrel{iid}{\sim} 0.5 N(\mathbf{0}, I_p) + 0.5 N(\mu, I_p)$, $x_{2i} \stackrel{iid}{\sim} 0.5 N(\mathbf{0}, I_p) + 0.5 N(\mathbf{0}, 1.5I_p)$ where $\mu=(0.5,0.5,-0.5,-0.5,0)^{\t}$, and $\epsilon_{1i}, \epsilon_{2i} \stackrel{iid}{\sim} t(5)$, the student's $t$-distribution with 5 degrees of freedom, independent of the features. Set $\alpha_1=\alpha_2=0$ under the null and $\alpha_1=0, \alpha_2=0.5$ under the alternative.
	\vspace{0.05in}
	
	{\bf Model C} (Gaussian mixture, additive spline, heteroskedastic). Let $y_{\ell i}= \theta(x_{\ell i}) + \epsilon_{\ell i}, i=1,\dots,n_{\ell}, \ell=1,2$, where $\theta(x) = \mathbb E(y|x)$ is an additive function of B-splines of covariates, $x_{1i} \stackrel{iid}{\sim} 0.5 N(\mathbf{0}, I_p) + 0.5 N(\mu, I_p)$, $x_{2i} \stackrel{iid}{\sim} 0.5 N(\mathbf{0}, I_p) + 0.5 N(\mathbf{0}, 1.5I_p)$ where $\mu=(0.5,0.5,-0.5,-0.5,0)^{\t}$. Set $\epsilon_{\ell i} \stackrel{}{\sim} N(0, 4/(1+x_{\ell i}^2(1))), \ell=1,2$, under the null and $\epsilon_{1i} \stackrel{}{\sim} N(0, 4/(1+x_{1i}^2(1)))$, $\epsilon_{2i} \stackrel{}{\sim} N(0, 2/(1+x_{2i}^2(1)))$ under the alternative. Here the noises are not independent of the covariates.
	\vspace{0.05in}
	
	In order to make density ratio estimation stable, we remove sample points whose marginal density ratio $f_{1,X}(x)/f_{2,X}(x)$ or the joint density ratio $f_1(x,y)/f_2(x,y)$ are outside of the interval $[1/100,100]$.
	
	\subsection{The low dimensional case}
	
	We first consider low-dimensional cases with $p=5$, setting the entries of $\beta$ to $\pm 1$ with random signs in Models A and B. In Model C, we multiply the coefficients to the B-spline transformation of predictors: $\theta(x) = \sum_{j=1}^5 \sum_{l=1}^4 \beta_{jl} b_l(x(j))$, where $b_l$'s are B-spline functions and $\beta_{jl}=\pm 1$ with randomly chosen signs.
	We consider sample sizes $n_1=n_2\in\{200, 500, 1000, 2000\}$ and randomly split each sample into two equal-sized subsets for the fitting and ranking steps.  Results for other split ratios are deferred to \Cref{sec:other_split_ratio}.
	In addition, we use four different probabilistic classification methods including linear logistic (LL), quadratic logistic (QL), neural network (NN) and kernel logistic regression (KLR).  Let $L$ be the class label such that $L_{1i}=1$ for $i\in \mathcal I_{12}$ and $L_{2i}=0$ for $i \in \mathcal I_{22}$. The KLR method \citep{zhu2005kernel} learns a kernel logistic regression classifier by minimizing
	$$ -\sum_{k\in\{1,2\}}\sum_{i \in \mathcal{I}_{k2}}\left[ L_{ki} \theta(x_{ki}; \boldsymbol \beta) - \log\{1+\exp(\theta(x_{ki};\boldsymbol \beta))\}\right] + \frac{\lambda}{2} \|\theta\|_{\mathcal{H}_K}^2 \,, $$
	where $\mathcal{H}_K$ is the reproducing kernel Hilbert spaces (RKHS) generated by the kernel $K(x,y) = \exp(-\|x-y\|^2/\sigma^2)$, $\theta(x;\boldsymbol \beta) = \beta_0 + \sum_{k\in\{1,2\}}\sum_{i \in \mathcal{I}_{k2}} \beta_{ki} K(x_{ki}, x)$ and $\eta(x;\boldsymbol \beta)=P(L=1|x) = 1/[1+\exp\{-\theta(x;\boldsymbol \beta)\}]$. Then we obtain the marginal ratio estimator $\hat{g}(x) = n_{12}\{1-\eta(x;\hat{\boldsymbol \beta})\}/\left\{n_{22}\eta(x;\hat{\boldsymbol \beta})\right\}$. 
	The joint classifier is obtained similarly using $(x_i,y_i)$.
	For the KLR method, the tuning parameter $\sigma^2=200$ is used in all settings. The more important tuning parameter $\lambda$ is chosen by minimizing the out-of-sample cross entropy loss. For large sample sizes such data-driven tuning is time-consuming to run for all repetitions. To reduce the overall running time for large sample sizes ($n_2=1000,~2000$), we use data-driven out-of-sample cross entropy loss to select $\lambda$ in first 20 repetitions, and then use the median of these 20 $\lambda$ values for the rest of the simulation. 
	Specifically, when $n_2=1000,2000$, we use $\lambda=0.05$ for both joint and marginal ratio estimates in Model A, $\lambda=0.0002$ for joint estimates and $\lambda=0.015$ for marginal estimates in Model B, $\lambda=0.015$ for joint estimates and $\lambda=0.02$ for marginal estimates in Model C.
	For the neural network method, we use the sigmoid activation function and the stochastic gradient descent algorithm, and two hidden layers are used in all considered models. 
	Moreover, we compare with the parametric likelihood ratio test (LRT) for hypotheses $H_0: \alpha_1 = \alpha_2$ versus $H_1: \alpha_1 \neq \alpha_2$ under the model $(Y_{\ell} | X) \sim N(\alpha_{\ell} + \beta^{\t}X, \sigma^2), \ell=1, 2$. Such a model is correctly specified under model A, where we expect LRT to have the best performance.  Under models B, C, this is misspecified.
	The code for the simulation is available to ensure reproducibility.

Among the four classification methods and three model settings considered, LL and QL rely on parametric models, which is correctly specified under model A, but incorrectly specified under models B and C, where the marginal distributions of $X$ are Gaussian mixtures, and the corresponding density ratios cannot be expressed as a logit transform of second order polynomials of $X$.  Moreover, the noise distribution is non-Gaussian under model B, and Gaussian but heteroskedastic in model C. The KLR and NN methods do not rely on any parametric model specification.  We summarize the model specification scenarios in \Cref{tab:model_specification} below.
\begin{table}[!ht]
	\centering
	\caption{\label{tab:model_specification} The model misspecification for all methods: ``\checkmark'' means the estimator uses a correctly specified parametric model to estimate the (conditional) density ratio; ``\ding{55}'' means that the estimator uses a misspecified parametric model to estimate the (conditional) density ratio; ``$-$'' means that the estimator is non-parametric and does not rely on any parametric model specification.}
	\begin{tabular}{|c|ccccc|ccccc|}
		\hline
	Task	& \multicolumn{5}{c|}{$\hat g$ } &\multicolumn{5}{c|}{$\hat v$}\\
	Method  & LRT& LL & QL & KLR & NN         & LRT& LL & QL & KLR & NN \\ 
	\hline
	Model A & $\checkmark$ &$\checkmark$ &\checkmark & $-$ & $-$ &$\checkmark$ &$\checkmark$& $\checkmark$ & $-$ & $-$\\
	Model B & \ding{55} &\ding{55} &\ding{55} & $-$ & $-$ &\ding{55} &\ding{55}& \ding{55} & $-$ & $-$\\
	Model C & \ding{55} &\ding{55} &\ding{55} & $-$ & $-$ &\ding{55} &\ding{55}& \ding{55} & $-$ & $-$\\
\hline
	\end{tabular}
\end{table}

	\begin{table}[!ht]
		\centering
		\caption{\label{tab:sim}Percentage of rejections over 500 repetitions using methods LL, QL, NN and KLR under the split ratio $0.5$ and the likelihood ratio test (LRT) with $\alpha=0.05$.}
		\begin{tabular}{|c|c|ccccc|ccccc|}
			\hline
			\multicolumn{2}{|c|}{\multirow{2}{*}{}} & \multicolumn{5}{c|}{Null} & \multicolumn{5}{c|}{Alternative} \\ \cline{3-12}
			\multicolumn{2}{|c|}{} & LL & QL & NN & KLR & LRT & LL & QL & NN & KLR & LRT \\ 
			\hline
			\multirow{4}{*}{Model A} & $n_2=200$ & 0.038 & 0.022 & 0.066 & 0.038 & 0.054 & 0.354 & 0.082 & 0.416 & 0.406 & 0.992 \\ 
			& $n_2=500$ & 0.044 & 0.020 & 0.058 & 0.066 & 0.054 & 0.654 & 0.392 & 0.644 & 0.698 & 1.000 \\ 
			& $n_2=1000$ & 0.042 & 0.038 & 0.058 & 0.056 & 0.032 & 0.898 & 0.770 & 0.866 & 0.912 & 1.000 \\ 
			& $n_2=2000$ & 0.048 & 0.056 & 0.068 & 0.068 & 0.038 & 0.980 & 0.974 & 0.990 & 0.990 & 1.000 \\ \hline
			\multirow{4}{*}{Model B} & $n_2=200$ & 0.044 & 0.058 & 0.052 & 0.062 & 0.054 & 0.200 & 0.116 & 0.180 & 0.224 & 0.236 \\ 
			& $n_2=500$ & 0.058 & 0.046 & 0.076 & 0.062 & 0.040 & 0.456 & 0.268 & 0.470 & 0.506 & 0.478 \\ 
			& $n_2=1000$ & 0.056 & 0.074 & 0.024 & 0.024 & 0.044 & 0.726 & 0.586 & 0.722 & 0.802 & 0.764 \\ 
			& $n_2=2000$ & 0.066 & 0.062 & 0.038 & 0.026 & 0.080 & 0.946 & 0.936 & 1.000 & 0.994 & 0.970 \\ \hline 
			\multirow{4}{*}{Model C} & $n_2=200$ & 0.080 & 0.066 & 0.056 & 0.064 & 0.090 & 0.074 & 0.204 & 0.062 & 0.300 & 0.088 \\ 
			& $n_2=500$ & 0.070 & 0.036 & 0.066 & 0.040 & 0.096 & 0.080 & 0.624 & 0.504 & 0.796 & 0.116 \\ 
			& $n_2=1000$ & 0.098 & 0.046 & 0.030 & 0.064 & 0.142 & 0.088 & 0.964 & 0.900 & 0.988 & 0.144 \\ 
			& $n_2=2000$ & 0.110 & 0.050 & 0.080 & 0.064 & 0.182 & 0.094 & 0.998 & 1.000 & 1.000 & 0.176 \\ \hline 
		\end{tabular}
	\end{table}

	All the simulation results in the following are computed over 500 repetitions with nominal type I error level $\alpha = 0.05$. 
	The results are summarized in Table \ref{tab:sim}.
	In Model A, the LL estimator is the correctly specified parametric method, and has very good performance even with small sample sizes.
	QL has relatively inaccurate empirical size for very small sample size $n_2=200$, because this model involves more parameters including all the linear and quadratic terms.
	The NN and KLR estimators are fully nonparametric, but still yield satisfying control of the type I error even under moderate sample sizes. 
	When the alternative hypothesis is true, the power increases as the sample size increases.
	The NN and KLR methods yield comparable or larger power against the LL approach even when the sample size is small, thanks to accurately estimated density ratios $v$ and $g$. The QL estimator requires a larger sample size to obtain reasonable power.  The performance using the estimated $\hat v, \hat g$ is indeed close to using the true functions $v, g$, which is presented in Table \ref{tab:sima}.
	
	Models B and C represent more challenging scenarios where the the marginal and joint ratios are nonparametric. The parametric methods such as LL and QL fail to yield correct type I error control or have limited ability to capture the difference under the alternative in at least one model setting, which demonstrates the limitation of such parametric approaches. The nonparametric methods NN and KLR both provide robust type I and II error control, with empirical type I error rate close to the nominal level.
	Moreover, the simulation results demonstrate that the data-driven tuning by minimizing out-of-sample cross entropy can have good practical performance.
	
	Note that in Model A, $\alpha_1 = \alpha_2$ under the null and $\alpha_1 \neq \alpha_2$ under the alternative hypothesis, while in Models B and C, the parametric assumption for the LRT is violated.
	As shown in \Cref{tab:sim}, the LRT method performs the best in Model A, which makes sense since it maximizes the usage of known information about the data and does not require sample splitting. In Model B, the LRT is not correctly specified because the true conditional distribution of $Y$ given $X$ is $t$-distribution and the conditional mean is not a linear function. Although it shows some ability to capture the distributional difference, it also has inflated type I error due to model misspecification. Further, the LL, KLR, and NN have comparable power with the LRT, demonstrating the advantage of the proposed approach by efficiently aggregating multiple conformal $p$-values, which reduces the impact of the sample splitting.  Moreover, in Model C, the LRT method fails to control the type I error, and is significantly less powerful than the QL, NN and KLR methods, demonstrating the benefit brought by the flexibility of our method in choosing the estimation algorithms in nonparametric models.
	
	\subsection{The high dimensional case}\label{sec:sim_hi}
	The flexibility of choosing probabilistic classification algorithms makes our method applicable in high-dimensional problems. Here we illustrate its performance in a high-dimensional scenario, which is similar to Model A in the low dimensional case but with ambient dimensionality $p=500$ and signal dimensionality $s=5$. The additional coordinates of $X$ are generated by iid standard normal, and the corresponding coefficients of $\boldsymbol\beta$ are filled with zeros. 
	Here we focus on a sparse linear classifier and investigate the effect of tuning and regularization. Letting $L$ be the class label with $L=1$ for training data and $L=0$ for testing data, we learn a sparse logistic regression model by minimizing
	$$ -\frac{1}{n_{12} + n_{22}}\sum_{k \in \{1, 2\} }\sum_{i \in \mathcal{I}_{k2}}\left[ L_{ki} \log \eta(x_{ki}; \boldsymbol \beta) + (1-L_{ki}) \log\{1-\eta(x_{ki};\boldsymbol \beta)\}\right] + \lambda \|\boldsymbol \beta\|_1 \,, $$
	where $\boldsymbol \beta = (\beta_0, \beta_1, \dots, \beta_p)$ and $\eta(x;\boldsymbol \beta) = P(L=1|x) = 1/[1+\exp\{-\beta_0 - \sum_{j=1}^p \beta_j x(j)\}]$. Then we obtain the marginal ratio estimator with $\hat{g}(x) = n_{12}\{1-\eta(x;\hat{\boldsymbol \beta})\}/\left\{n_{22}\eta(x;\hat{\boldsymbol \beta})\right\}$. In a similar manner, we can the estimate the joint density ratio and hence the conditional density ratio.
	
	We consider sample sizes $n_1= n_2 = 1000$ and 2000. The empirical rejection frequency and estimation errors of the conformal weights, $Err_{\hat{g}} = n_{11}^{-1}\sum_{i} |\hat G_{1i}/\sum_{i} \hat G_{1i} - G_{1i}/\sum_{i} G_{1i} |$, are shown in \Cref{fig:higha}.  Since the estimation errors are not observable in practice, we plot the out-of-sample marginal cross entropy error (MCEntropy, defined as $-L\log\hat p-(1-L)\log(1-\hat p)$) in the classification problem involved in estimating $\hat g$ (the solid lines with star-shaped marks).
	Moreover, under the alternative, we also report the empirical out-of-sample estimation error of the
	conditional density ratio $v$, defined by $Err_{\hat{v}} = (n_{11}+n_{21})^{-1}\left\{\sum_{i}(\hat{V}_{1i} - V_{1i})^2 + \sum_{j} (\hat{V}_{2j} - V_{2j})^2\right\}$.  
	
	Again, when the true marginal density ratios are used, the empirical sizes are close to the nominal level $\alpha=0.05$ as expected. 
	When the marginal density ratios are estimated, the type I error is well controlled for a wide range of tuning parameter values, indicating good robustness of validity. The plot of out-of-sample marginal cross entropy error suggests that in practice one can choose the tuning parameter value near the elbow of the error plot.
	Under the alternative hypotheses, the power is maximized at tuning parameter values corresponding to the smallest estimation error in $\hat v$, which can also be chosen using its out-of-sample cross entropy error plot (not shown in the plots). Practically one can also use separate tuning parameters for the marginal classification and joint classification. 
	
	\begin{figure}
		\centering
		\newcommand{\thiswidth}{0.46\linewidth}
		\newcommand{\thisgap}{1mm}
		\begin{tabular}{cc}
			\hspace{\thisgap}\includegraphics[width=\thiswidth]{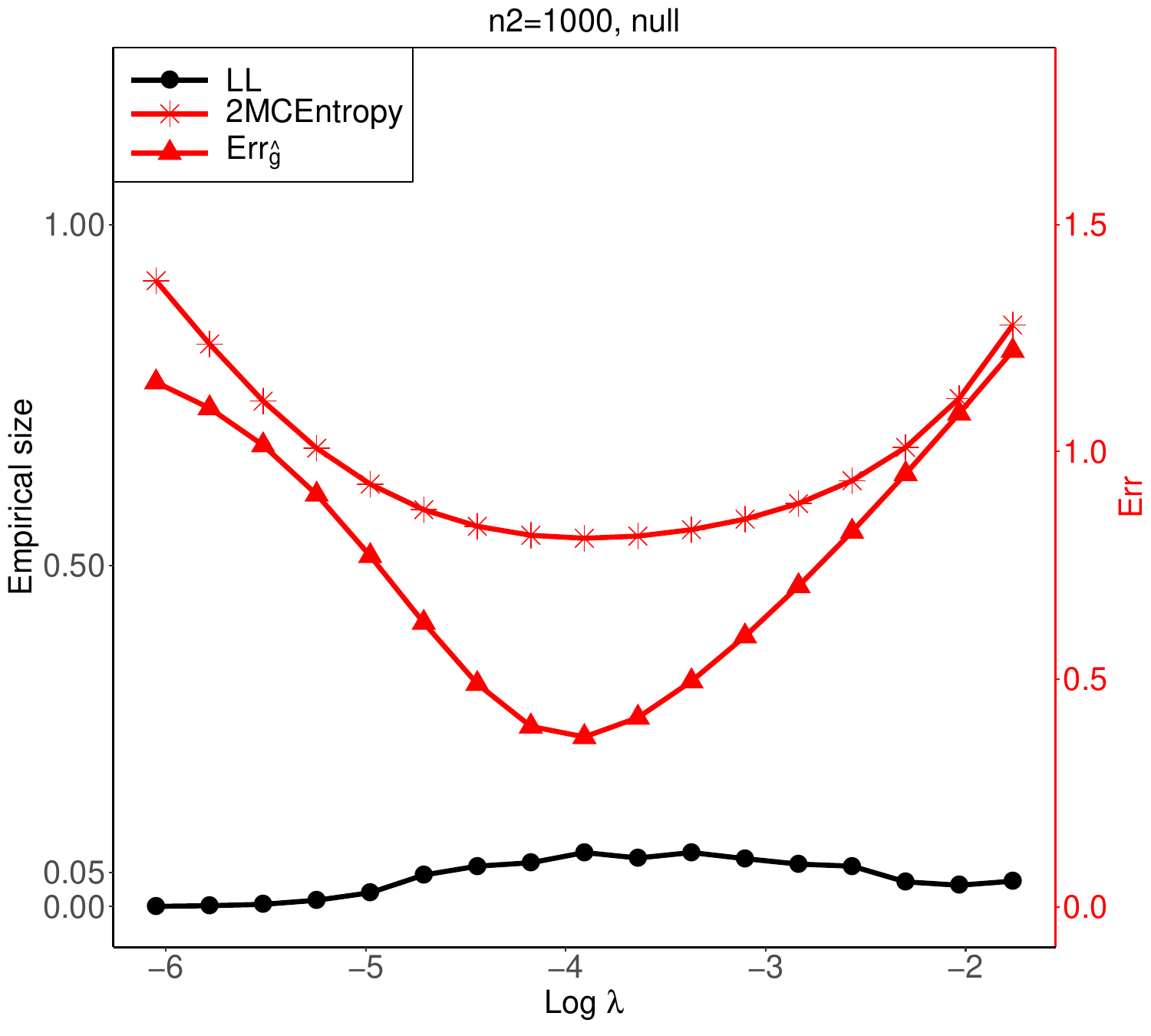} &
			\hspace{\thisgap}\includegraphics[width=\thiswidth]{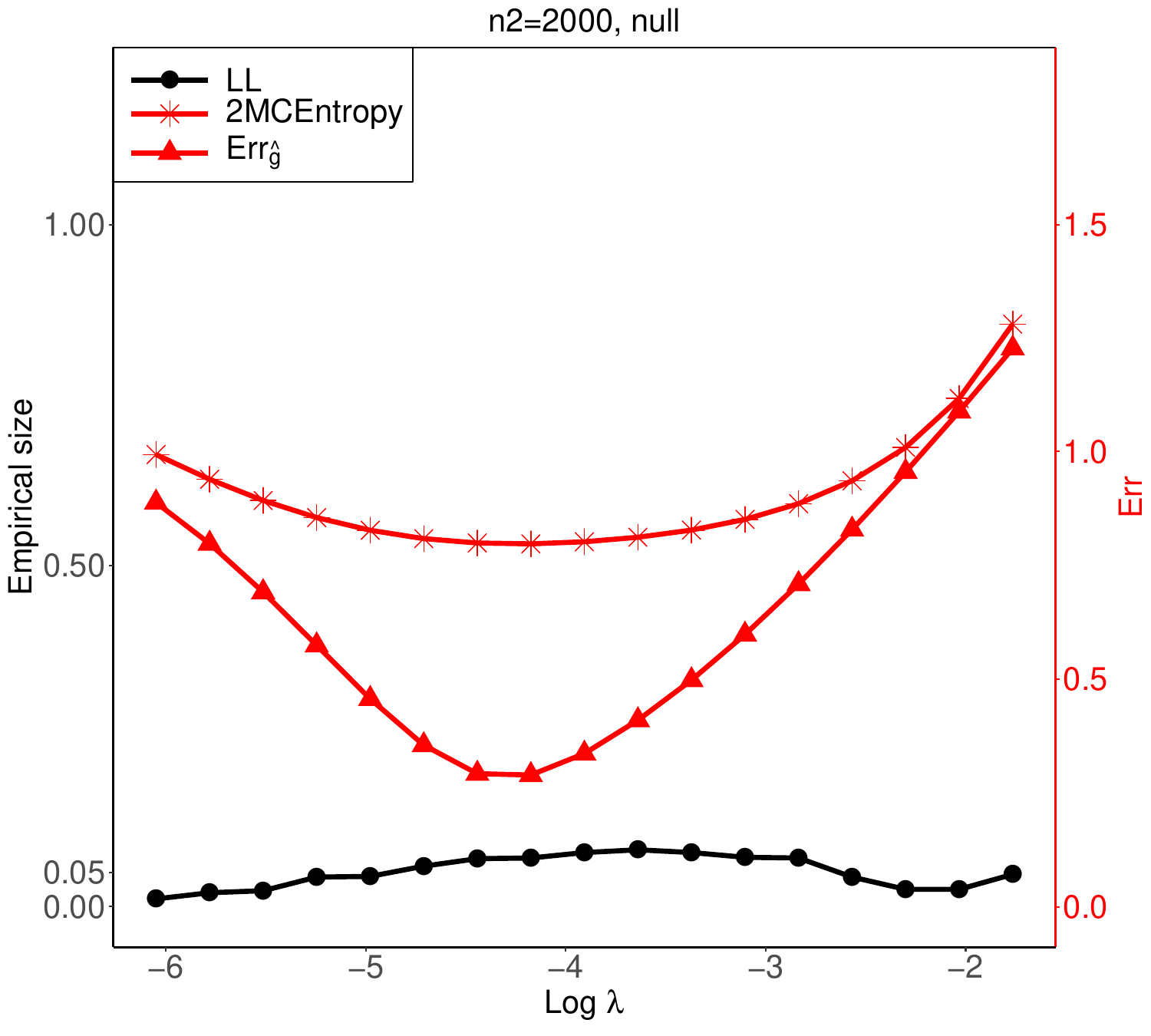}  \\ 
			\hspace{\thisgap}\includegraphics[width=\thiswidth]{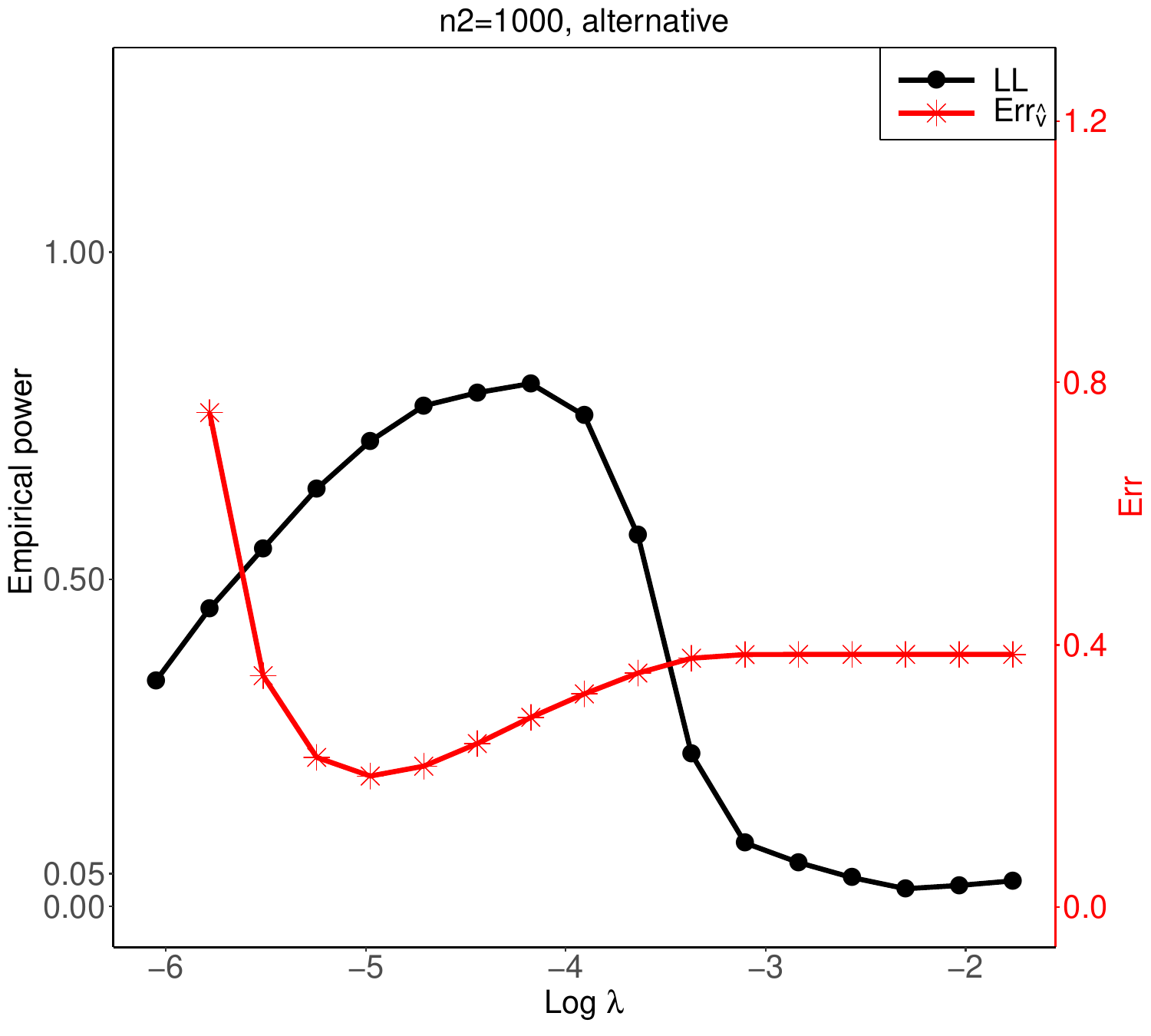} &
			\hspace{\thisgap}\includegraphics[width=\thiswidth]{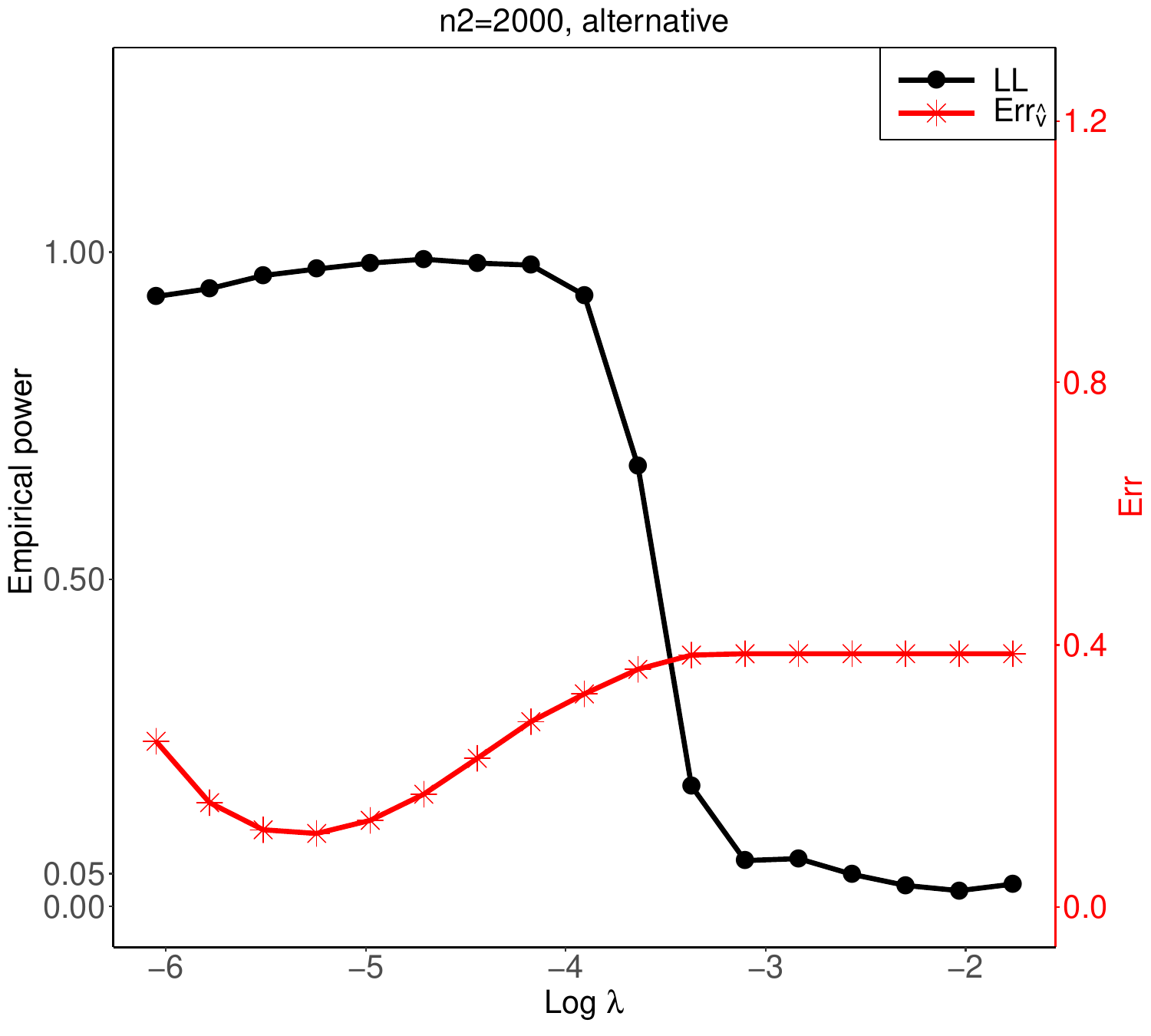} 		
		\end{tabular}
		\caption{Performance of empirical rejection frequency in Model A with $p=500, s=5$ under the null (top) and alternative (bottom) over 500 repetitions, across a variety of regularization parameters using sparse LL with $\alpha=0.05$ and split ratio 0.5.}
		\label{fig:higha}
	\end{figure}
	
	\section{A synthetic data example}\label{sec:real}
	
	We consider the airfoil data set from the UCI Machine Learning Repository \citep{Dua:2019}, which has $n = 1503$ observations of a response $Y$ (scaled sound pressure level of NASA airfoils), and a covariate $X$ with $p = 5$ dimensions (log frequency, angle of attack, chord length, free-stream velocity, and suction side log displacement thickness).  This data set has been used by \cite{tibshirani2019conformal}, who first studied weighted conformalization.
	
	The original data does not have a two-sample separation.
	We consider five experiments based on different ways to generate the two populations. 
	\begin{itemize}
		\item[(i)] Random partition. We randomly partition the data set with $n_1=751$ and $n_2=752$.
		\item[(ii)] Random partition and exponential tilting. We first randomly partition the data into two sets $\data_1$, $\tilde{\data}_2$. Then following \cite{tibshirani2019conformal}, we construct $\data_2$ by sampling $25\%$ of the points from $\tilde{\data}_2$ with replacement, with probabilities proportional to
		$$ w(x) = \exp(x^{\t}\alpha), ~~ \mathrm{where} ~~ \alpha = (-1, 0, 0, 0, 1) \,. $$
		The final sample sizes are $n_1=301$ and $n_2=301$.
		\item[(iii)] Chord-based partition. We split the data set into two subsets where the values of the ``chord'' variable in $\data_1$ are smaller than the 50\% quantile and exclude the ``chord'' variable in subsequent analyses, resulting in $n_1=778$ and $n_2=725$. To avoid singularity between the two populations, we randomly select $0.05 n_1$ samples in $\data_1$ and the same amount of samples in $\data_2$ to flip their groups.
		\item[(iv)] Velocity-based partition. We partition the data set into two subsets where the values of the ``velocity'' variable in $\data_1$ are smaller than the 50\% quantile and remove the covariate ``velocity'' from subsequent analyses, with $n_1=761$ and $n_2=742$. We randomly select $0.05 n_1$ samples in $\data_1$ and the same amount of samples in $\data_2$ to flip their groups.   
		\item[(v)] Response-based partition. We split the data according to the value of the response variable, where the first group contains the sample points with smaller response values, while $\data_2$ contains the rest, with $n_1=752$ and $n_2=751$.  A similar label flipping is applied to avoid singularity.
	\end{itemize}

\begin{table}
	\caption{\label{tab:airfoil}Percentage of rejections (PR) and average classification error of estimating $g$ in Airfoil data set for cases (i-ii) using methods LL and NN over 500 repetitions, with $\alpha=0.05$ and split ratio 0.5.}
	\centering
	\begin{tabular}{|c|c|c|c|c|}
		\hline
		& $\mathrm{PR_{LL}}$ & $\mathrm{PR_{NN}}$  & $\mathrm{MCE_{LL}}$ & $\mathrm{MCE_{NN}}$ \\ \hline
		case (i) & 0.048 & 0.056 &  0.500 & 0.500   \\ 
		\hline
		case (ii) & 0.052 & 0.068 & 0.208 & 0.208 \\ \hline
	\end{tabular}
\end{table}

\begin{table}
	\caption{\label{tab:airfoil2}Median p-values (Pval) and average classification error of estimating $g$ in Airfoil data set for cases (iii-v) using methods LL and NN over 500 random splits, with $\alpha=0.05$ and split ratio 0.5.} 
	\centering
	\begin{tabular}{|c|c|c|c|c|}
		\hline
		& $\mathrm{Pval_{LL}}$ & $\mathrm{Pval_{NN}}$  & $\mathrm{MCE_{LL}}$ & $\mathrm{MCE_{NN}}$   \\ \hline		
		case (iii) & 0.005 & 0.472 & 0.160 & 0.053  \\ \hline
		case (iv) & 0.000 & 0.412 & 0.444 & 0.066 \\ \hline
		case (v) & 0.000 & 0.002 & 0.233 & 0.136 \\ \hline
	\end{tabular}
\end{table}
	
	As in the simulation study, we split each group into two equal-sized subsets, and conduct the two-sample conditional distribution test at significance level $\alpha=0.05$. Linear logistic regression (LL) and neural network (NN) are used to estimate the density ratios. 
	The neural network uses two hidden layers with ten nodes in all cases.
	Each experiment is repeated for 500 trials.
	
	For experiments (i) and (ii), which are clearly under the null hypothesis, we can re-generate the data by repeating the random generation of the two subsamples. With these repeatedly generated data sets, we can compute the empirical frequency of rejections. 
	As shown in Table \ref{tab:airfoil}, the type I errors are close to the nominal level. 
	
	For experiments (iii-v), we can only have a single deterministic generation of the training and testing data, except a small fraction of group flipping, and only a single $p$-value can be computed.  We use these experiments to illustrate the effect of multiple realizations of auxiliary randomization.  Recall that our method uses auxiliary randomization to split the data sets into fitting and ranking subsamples.  Such auxiliary randomization may lead to different results on the same data set if the inference is carried out independently by different researchers.  To mitigate this effect, one can obtain multiple $p$-values using multiple realizations of auxiliary randomization.  Although each single $p$-value has asymptotically valid null distribution, their dependence requires a careful aggregation of these $p$-values.  Here we use a median $p$-value approach in \cite{DiCiccioDR20}. Formally, suppose we repeat the auxiliary randomization $B$ times, obtaining $p$-values $\hat p_1,...,\hat p_B$.
	Then 
	$\hat p=1\wedge \left[2\times {\rm Median}(\hat p_1,...,\hat p_B)\right]$
	is a valid $p$-value.
	
	In all experiments (i-v), we use out-of-sample marginal classification error (MCE) as a proxy of the accuracy of marginal density ratio estimation. 
	In experiment (iii), both LL and NN methods give large $p$-values, and the NN method also has small marginal classification errors.  Thus there is no strong evidence against the covariate shift assumption.
	In experiment (iv), the LL method gives small $p$-values while the NN method suggests otherwise.  But the marginal classification errors indicate that the NN method is likely to be more accurate in estimating the marginal density ratios, and hence provides more trustworthy $p$-values.
	In experiment (v), both methods agree to reject the null hypothesis, which is the correct decision by the construction of training and testing samples. The neural net method also gives marginal classification errors comparable to those in the null cases, further confirming the validity of $p$-value.

	\section{Discussion}\label{sec:discussion}
	In applications it is often the case that the training data $(X_{1i},Y_{1i})_{i=1}^{n_1}$ has a large sample size, whereas the testing data $(X_{2j},Y_{2j})_{j=1}^{n_2}$ has a limited sample size. As our theory and experiments have demonstrated, a valid type I error control of the proposed test only depends on the accuracy of marginal density ratio estimation.  Our method would be particularly useful in the semi-supervised scenario, where unlabeled testing sample points $X_{2j}$ are easy to obtain.  In this case we can use these unlabeled testing sample points to estimate the marginal density ratio, and save the scarce labeled testing sample to estimate the joint density ratio.
	
	
	The use of sample splitting and auxiliary randomization for valid and efficient statistical inference has been studied by many authors in the high dimensional regression literature \citep{wasserman2009high,meinshausen2009p,rinaldo2019bootstrapping}, and more recently in the conformal inference literature \citep{kuchibhotla2019nested,kim2020predictive}. The theory in \cite{kim2020predictive} also required an inflated non-coverage by a factor of two after aggregating multiple subsamples.  It is unclear whether such a loss of coverage is unavoidable. It would be interesting and important to better understand the dependence between the $p$-values from multiple splits, and improve the current conservative inflation method when combining multiple $p$-values.

	Conformal methods are initially developed without sample splitting, but in a leave-one-out manner.  The validity of such conformal $p$-values comes from the symmetry among the data points.  A straightforward leave-one-out version of our method would be to leave out each pair of sample points $(X_{1i},Y_{1i})$, $(X_{2j},Y_{2j})$ for $i\in[n_1]$ and $j\in [n_2]$, and fit $\hat g$, $\hat v$ using all remaining data.  Such an implementation will require $n_1 n_2$ re-fitting of $\hat g$, $\hat v$, which is computationally prohibitive. Some efficient implementation of leave-one-out update or warm start techniques would be necessary. Also, the complex dependencies among the resulting conformity scores bring further theoretical challenges in understanding the aggregated conformal $p$-values.  Given the potential improved sample efficiency, these questions may be investigated in future works.
	\bibliographystyle{asa}
	\bibliography{test}
	
	\newpage
	
	\appendix
	
	\section{Asymptotic variance using importance sampling}\label{sec:asy_var_aux}
	We illustrate the importance sampling idea for estimating the asymptotic variance using the samples from $P_2$.  For simplicity we focus on the ideal statistic $T$ which uses the true density ratio functions.  The same idea can be directly carried over to the actual statistic that uses estimated density ratios.
	
	Recall that the asymptotic variance of $\sqrt{n_{11}}n_{21}^{-1}\sum_j U_j$ is given by $\sigma^2 = \sigma_1^2 + n_{11}/(12n_{21}) + \sigma_2^2/4 - \rho_{12} $ where $\sigma_1^2 = \mathrm{Var}\left[ G_{11}\{1-F_{1/2}(V_{11})\} \right]$, $\sigma_2^2 = \mathrm{Var}(G_{11})$ and $\rho_{12} = \mathrm{Cov}[G_{11}\{1-F_{1/2}(V_{11})\}, G_{11} ]$.
	
	Note that under $H_0$,
	\begin{align*}
	\sigma_1^2 = & \int \frac{f_2^2(z)}{f_1^2(z)}\{1-F_{1/2}(v(z))\}^2 f_1(z)dz - \frac{1}{4}\\
	=& \int \frac{f_2(z)}{f_1(z)}\{1-F_{1/2}(v(z))\}^2 f_2(z)dz-\frac{1}{4}\\
	=& \mathbb E G_{21}\{1- F_{1/2}(V_{21})\}^2-\frac{1}{4}\,.
	\end{align*}
	
	Analogously, we have
	$$ \sigma_2^2 = \mathbb E G_{21} -1 \,,$$
	$$ \rho_{12} = \mathbb E G_{21}\{1-F_{1/2}(V_{21})\} - 1/2 \,.  $$
	
	Essentially, the change of base measure allows us to represent the expectations under $P_1$ to corresponding expectations under $P_2$. So we can use the sample $(X_{2j},Y_{2j})_{j=1}^{n_{21}}$ to estimate the asymptotic variance.  The consistency of the importance sample estimate follows the same strategy as the proof for the original estimate and is omitted.
	
	\section{More simulation results}
	
	\subsection{Additional simulation results under different splitting ratios}\label{sec:other_split_ratio}
	
	To investigate the effect of the splitting ratio on the performance of the test, we consider $r=0.3, 0.5$ and 0.8, respectively, where $n_{11} =\lceil n_1*r \rceil$ and $n_{21} = \lceil n_2 * r \rceil$. We provide the results for Models A, B and C in Tables \ref{tab:sima}-\ref{tab:simc}. Additionally, the results of the proposed approach with the true functions $g, v$ are also included in these tables, denoted by ``Oracle". We should note that the ``Oracle" requires no estimation, thus it achieves higher power as $r$ increases while still controlling the type I error.
	By contrast, when $r$ increases, the sample size of the fitting data decreases, and the estimation algorithms may produce estimates with limited accuracy for the testing procedure. For example, as shown in Table \ref{tab:sima}, the NN fails to control the type I error in model A under $n_2=200$ and $r=0.8$. Moreover, a larger splitting ratio $r$ may lead to lower power due to the inaccurate estimates. From Tables \ref{tab:simb} and \ref{tab:simc}, when $n_2=2000$, the NN achieves lower power under $r=0.8$ than $r=0.5$.
	In contrast, a smaller $r$ will lead to more accurate estimates, but the test may suffer from power loss because less data are used for constructing the test statistic.
	More specifically, when $r=0.3$, though the type I error is well controlled for NN and KLR in all considered models, the power is lower than $r=0.5$ especially for the small sample size.
	Overall, we suggest using $r=0.5$, and a smaller $r$ is allowed if one is not very confident in the resulting estimators in practice.
	
	\begin{table}[!ht]
		\centering
		\caption{\label{tab:sima}Percentage of rejections using methods LL, QL, NN and KLR with $\alpha=0.05$ under different sample splitting ratios for Model A.}
		\begin{tabular}{|c|c|ccccc|ccccc|}
			\hline
			\multicolumn{2}{|c|}{\multirow{2}{*}{}} & \multicolumn{5}{c|}{Null} & \multicolumn{5}{c|}{Alternative} \\ \cline{3-12}
			\multicolumn{2}{|c|}{} & LL & QL & NN & KLR & Oracle & LL & QL & NN & KLR & Oracle \\ 
			\hline
			\multirow{3}{*}{$n_2=200$} & r=0.3 & 0.024 & 0.024 & 0.034 & 0.046 & 0.042 & 0.282 & 0.092 & 0.338 & 0.336 & 0.300 \\ 
			& r=0.5 & 0.038 & 0.022 & 0.066 & 0.038 & 0.052 & 0.354 & 0.082 & 0.416 & 0.406 & 0.444 \\ 
			& r=0.8 & 0.034 & 0.010 & 0.136 & 0.042 & 0.042 & 0.350 & 0.010 & 0.398 & 0.408 & 0.578 \\ \hline
			\multirow{3}{*}{$n_2=500$} & r=0.3 & 0.050 & 0.026 & 0.034 & 0.052 & 0.038 & 0.550 & 0.382 & 0.496 & 0.534 & 0.580 \\ 
			& r=0.5 & 0.044 & 0.020 & 0.058 & 0.066 & 0.038 & 0.654 & 0.392 & 0.644 & 0.698 & 0.720 \\ 
			& r=0.8 & 0.032 & 0.024 & 0.086 & 0.042 & 0.066 & 0.770 & 0.216 & 0.822 & 0.752 & 0.864 \\ \hline
			\multirow{3}{*}{$n_2=1000$} & r=0.3 & 0.048 & 0.044 & 0.068 & 0.044 & 0.024 & 0.744 & 0.672 & 0.730 & 0.770 & 0.754 \\ 
			& r=0.5 & 0.042 & 0.038 & 0.058 & 0.056 & 0.050 & 0.898 & 0.770 & 0.866 & 0.912 & 0.902 \\ 
			& r=0.8 & 0.048 & 0.030 & 0.088 & 0.070 & 0.054 & 0.962 & 0.700 & 0.932 & 0.974 & 0.980 \\ \hline
			\multirow{3}{*}{$n_2=2000$} & r=0.3 & 0.042 & 0.052 & 0.050 & 0.066 & 0.040 & 0.938 & 0.902 & 0.920 & 0.932 & 0.934 \\ 
			& r=0.5 & 0.048 & 0.056 & 0.068 & 0.068 & 0.070 & 0.980 & 0.974 & 0.990 & 0.990 & 0.986 \\ 
			& r=0.8 & 0.064 & 0.044 & 0.088 & 0.066 & 0.064 & 1.000 & 0.968 & 0.994 & 0.998 & 1.000  \\ \hline
		\end{tabular}
	\end{table}
	
	\begin{table}[!ht]
		\centering
		\caption{\label{tab:simb}Percentage of rejections using methods LL, QL, NN and KLR with $\alpha=0.05$ under different sample splitting ratios for Model B.}
		\begin{tabular}{|c|c|ccccc|ccccc|}
			\hline
			\multicolumn{2}{|c|}{\multirow{2}{*}{}} & \multicolumn{5}{c|}{Null} & \multicolumn{5}{c|}{Alternative} \\ \cline{3-12}
			\multicolumn{2}{|c|}{} & LL & QL & NN & KLR & Oracle & LL & QL & NN & KLR & Oracle \\ 
			\hline
			\multirow{3}{*}{$n_2=200$} & r=0.3 & 0.064 & 0.048 & 0.056 & 0.044 & 0.060 & 0.176 & 0.128 & 0.192 & 0.178 & 0.680 \\ 
			& r=0.5 & 0.044 & 0.058 & 0.052 & 0.062 & 0.066 & 0.200 & 0.116 & 0.180 & 0.224 & 0.858 \\ 
			& r=0.8 & 0.066 & 0.032 & 0.032 & 0.112 & 0.052 & 0.194 & 0.068 & 0.086 & 0.302 & 0.970  \\ \hline
			\multirow{3}{*}{$n_2=500$} & r=0.3 & 0.070 & 0.056 & 0.050 & 0.042 & 0.058 & 0.338 & 0.246 & 0.356 & 0.352 & 0.942  \\ 
			& r=0.5 & 0.058 & 0.046 & 0.076 & 0.062 & 0.044 & 0.456 & 0.268 & 0.470 & 0.506 & 0.994 \\ 
			& r=0.8 & 0.058 & 0.070 & 0.038 & 0.152 & 0.064 & 0.518 & 0.206 & 0.458 & 0.560 & 0.998  \\ \hline
			\multirow{3}{*}{$n_2=1000$} & r=0.3 & 0.052 & 0.054 & 0.030 & 0.030 & 0.056 & 0.584 & 0.514 & 0.714 & 0.728 & 1.000 \\ 
			& r=0.5 & 0.056 & 0.074 & 0.024 & 0.024 & 0.054 & 0.726 & 0.586 & 0.722 & 0.802 & 1.000 \\ 
			& r=0.8 & 0.062 & 0.066 & 0.180 & 0.024 & 0.052 & 0.774 & 0.390 & 0.774 & 0.608 & 1.000 \\ \hline
			\multirow{3}{*}{$n_2=2000$} & r=0.3 & 0.064 & 0.050 & 0.054 & 0.054 & 0.062 & 0.822 & 0.858 & 0.998 & 0.986 & 1.000 \\ 
			& r=0.5 & 0.066 & 0.062 & 0.038 & 0.026 & 0.060 & 0.946 & 0.936 & 1.000 & 0.994 & 1.000  \\ 
			& r=0.8 & 0.056 & 0.078 & 0.054 & 0.022 & 0.064 & 0.952 & 0.810 & 0.866 & 0.972 & 1.000 \\ \hline
		\end{tabular}
	\end{table}

	\begin{table}[!ht]
		\centering
		\caption{\label{tab:simc}Percentage of rejections using methods LL, QL, NN and KLR with $\alpha=0.05$ under different sample splitting ratios for Model C.}
		\begin{tabular}{|c|c|ccccc|ccccc|}
			\hline
			\multicolumn{2}{|c|}{\multirow{2}{*}{}} & \multicolumn{5}{c|}{Null} & \multicolumn{5}{c|}{Alternative} \\ \cline{3-12}
			\multicolumn{2}{|c|}{} & LL & QL & NN & KLR & Oracle & LL & QL & NN & KLR & Oracle \\ 
			\hline
			\multirow{3}{*}{$n_2=200$} & r=0.3 & 0.074 & 0.064 & 0.048 & 0.062 & 0.050 & 0.062 & 0.190 & 0.072 & 0.202 &  0.530 \\ 
			& r=0.5 & 0.080 & 0.066 & 0.056 & 0.064 & 0.046 & 0.074 & 0.204 & 0.062 & 0.300 & 0.760 \\ 
			& r=0.8 & 0.062 & 0.022 & 0.040 & 0.098 & 0.060 & 0.062 & 0.118 & 0.040 & 0.230 & 0.908 \\ \hline
			\multirow{3}{*}{$n_2=500$} & r=0.3 & 0.080 & 0.042 & 0.038 & 0.054 & 0.038 & 0.062 & 0.580 & 0.468 & 0.684 & 0.900  \\ 
			& r=0.5 & 0.070 & 0.036 & 0.066 & 0.040 & 0.036 & 0.080 & 0.624 & 0.504 & 0.796 & 0.984 \\ 
			& r=0.8 & 0.074 & 0.034 & 0.056 & 0.074 & 0.044 & 0.060 & 0.366 & 0.132 & 0.702 & 0.994  \\ \hline
			\multirow{3}{*}{$n_2=1000$} & r=0.3 & 0.092 & 0.044 & 0.032 & 0.056 & 0.034 & 0.070 & 0.906 & 0.866 & 0.950 & 0.994  \\ 
			& r=0.5 & 0.098 & 0.046 & 0.030 & 0.064 & 0.058 & 0.088 & 0.964 & 0.900 & 0.988 &  1.000 \\ 
			& r=0.8 & 0.116 & 0.052 & 0.102 & 0.066 & 0.052 & 0.090 & 0.800 & 0.626 & 0.972 & 1.000 \\ \hline
			\multirow{3}{*}{$n_2=2000$} & r=0.3 & 0.100 & 0.046 & 0.068 & 0.046 & 0.062 & 0.090 & 0.994 & 1.000 & 0.998 & 1.000 \\ 
			& r=0.5 & 0.110 & 0.050 & 0.080 & 0.064 & 0.044 & 0.094 & 0.998 & 1.000 & 1.000 & 1.000 \\ 
			& r=0.8 & 0.138 & 0.048 & 0.044 & 0.050 & 0.030 & 0.092 & 0.990 & 0.974 & 1.000 & 1.000  \\ \hline
		\end{tabular}
	\end{table}
	
	\subsection{Error quantities}\label{sec:error_quantities}
	
	Let $\rho_1 =|\mathrm{Corr}_*(\hat G_{11}-G_{11}, \hat D_{11})|$, $\rho_2 = |\mathrm{Corr}_*(G_{11}, \hat D_{11})| |\mathbb E_* (\hat G_{11}-G_{11})|/\|\hat G_{11}-G_{11}\|_{1,*} $ and $\rho_3=|\mathbb E_* G_{11}(\hat D_{11}-D_{11})|$.
	As discussed in Section \ref{sec:thm}, to control the type I error asymptotically, the $\rho_1$ and $\rho_2$ should be sufficiently small. To guarantee the consistency of the test, we require the $\rho_3$ to be small enough so it does not destroy the signal in the data.
	Since it is hard to verify if these quantities converges theoretically because of the complex interactions of the variables in the simulation studies, we provide the empirical estimation of the quantities for the LL and KLR, see Figure \ref{fig:error}. Note that $\rho_{1,LL}$ denotes the empirical estimate of $\rho_1$ with the estimators obtained by the LL. Other quantities are defined similarly.
	
	 In Model A, the LL leads to smaller estimation error compared to the KLR method, which makes sense since the LL is the correctly specified parametric method. The error quantities for both LL and KLR decreases to nearly zero as the sample size increases, which provides empirical support for the assumptions. In Model B, though the LL and KLR are not correct models, they perform decently. Under the alternative, the $\hat{\rho}_{3,LL}$ converges slower than $\hat{\rho}_{3,KLR}$, which agrees with the power performance that the power of the LL is a bit lower than the KLR for the large sample size.
	 In Model C, the $\hat{\rho}_{1,LL}$ is larger than $\hat{\rho}_{1,KLR}$ for large sample sizes, which explains why the empirical type I error of the LL is large. Under the alternative of Model C, $\hat{\rho}_{3,LL}$ is not convergent, and the error seems too large that it destroys the signal. It explains the phenomenon that the LL barely shows any power in this case.
	
	\begin{figure}
		\centering
		\newcommand{\thiswidth}{0.46\linewidth}
		\newcommand{\thisgap}{1mm}
		\begin{tabular}{cc}
			\hspace{\thisgap}\includegraphics[width=\thiswidth]{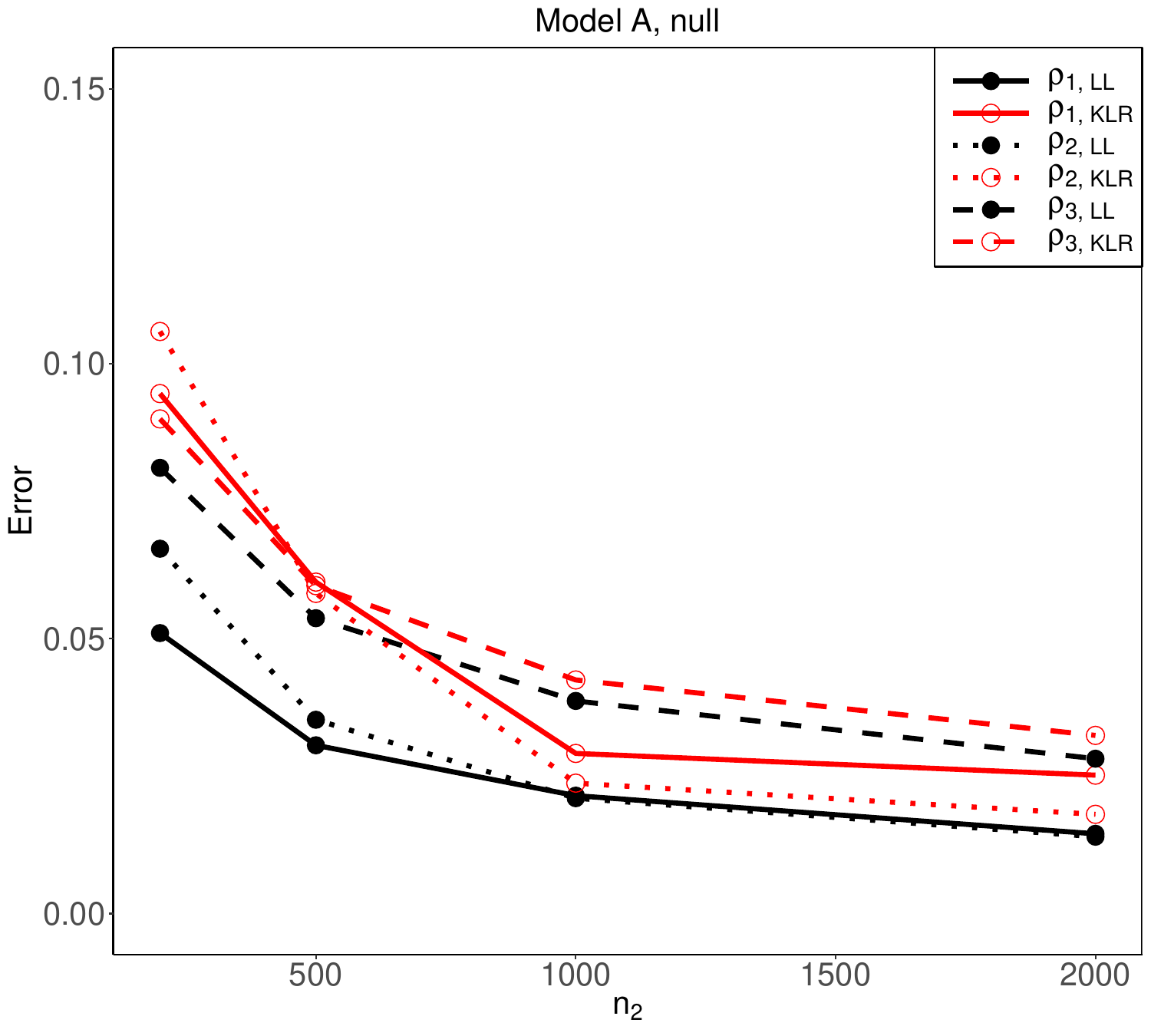} &
			\hspace{\thisgap}\includegraphics[width=\thiswidth]{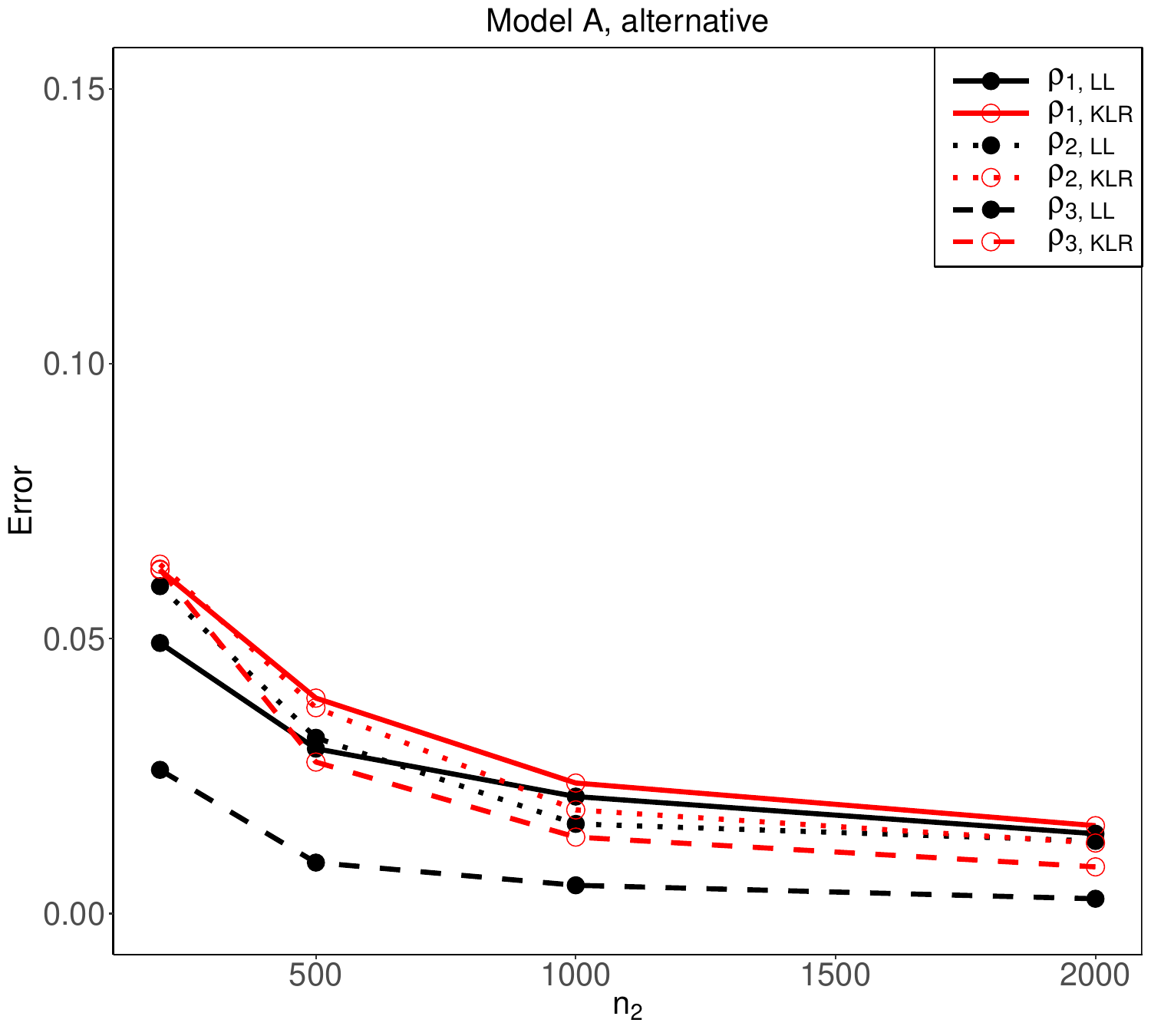}  \\ 
			\hspace{\thisgap}\includegraphics[width=\thiswidth]{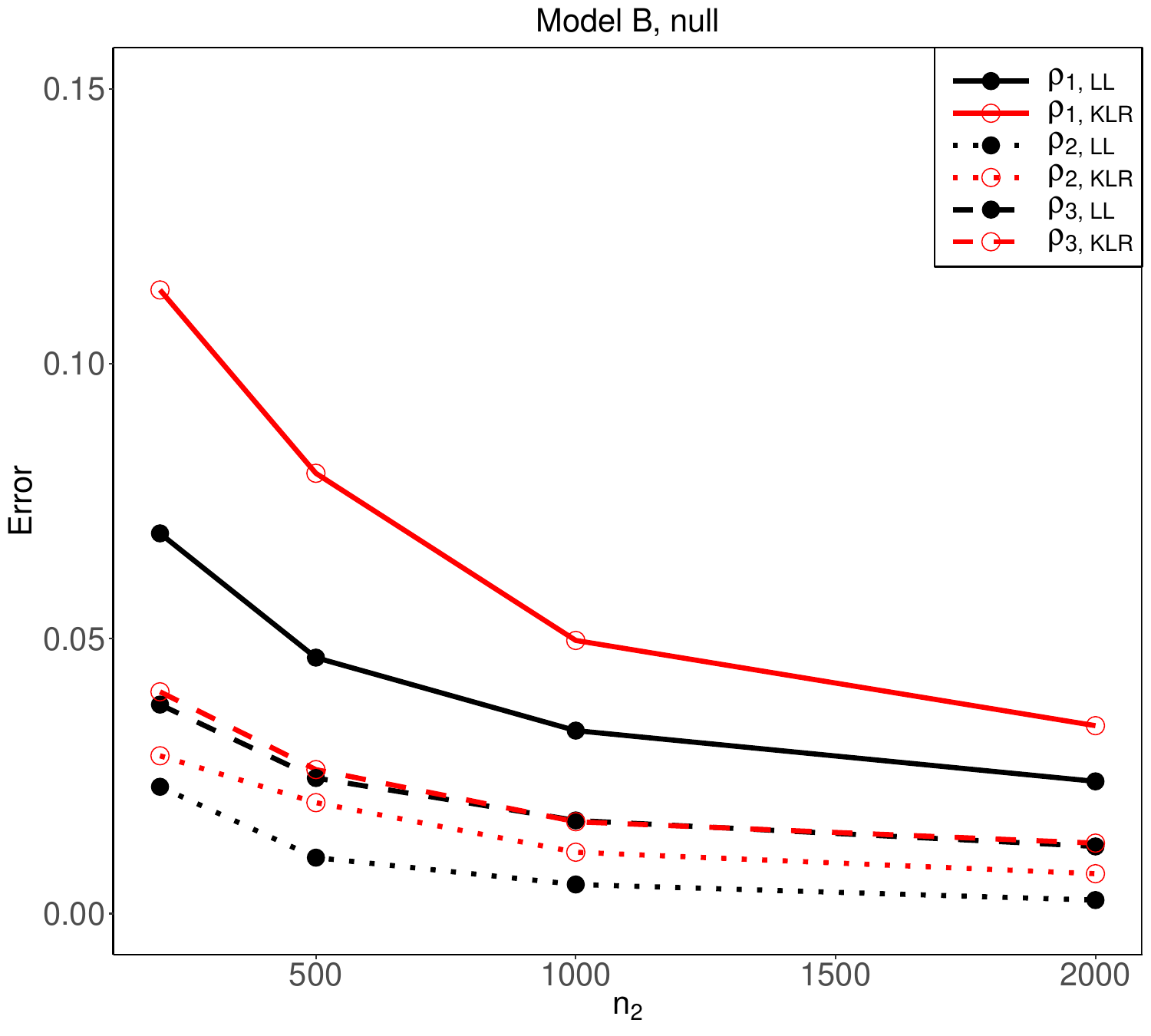} &
			\hspace{\thisgap}\includegraphics[width=\thiswidth]{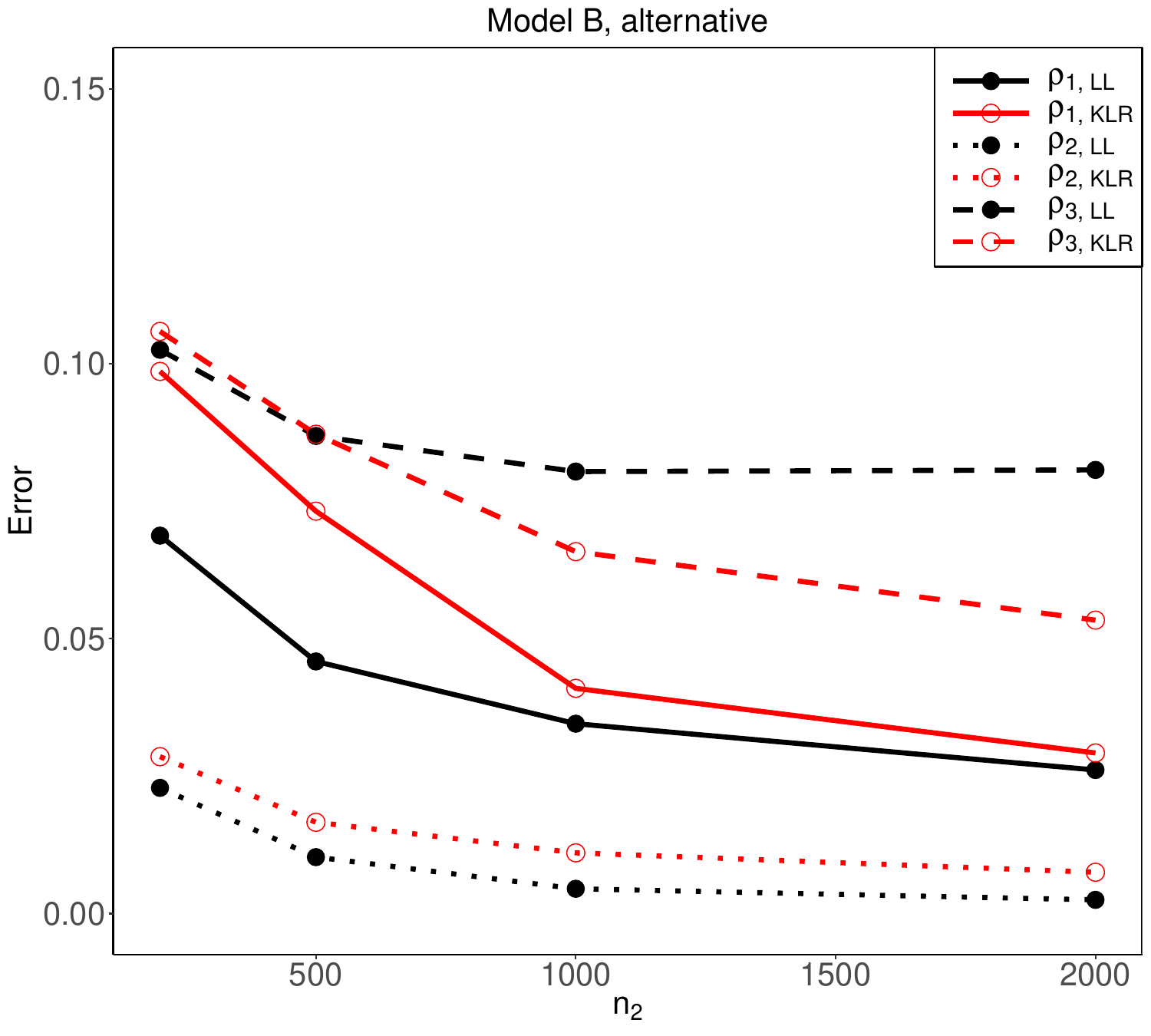} \\ 	
			\hspace{\thisgap}\includegraphics[width=\thiswidth]{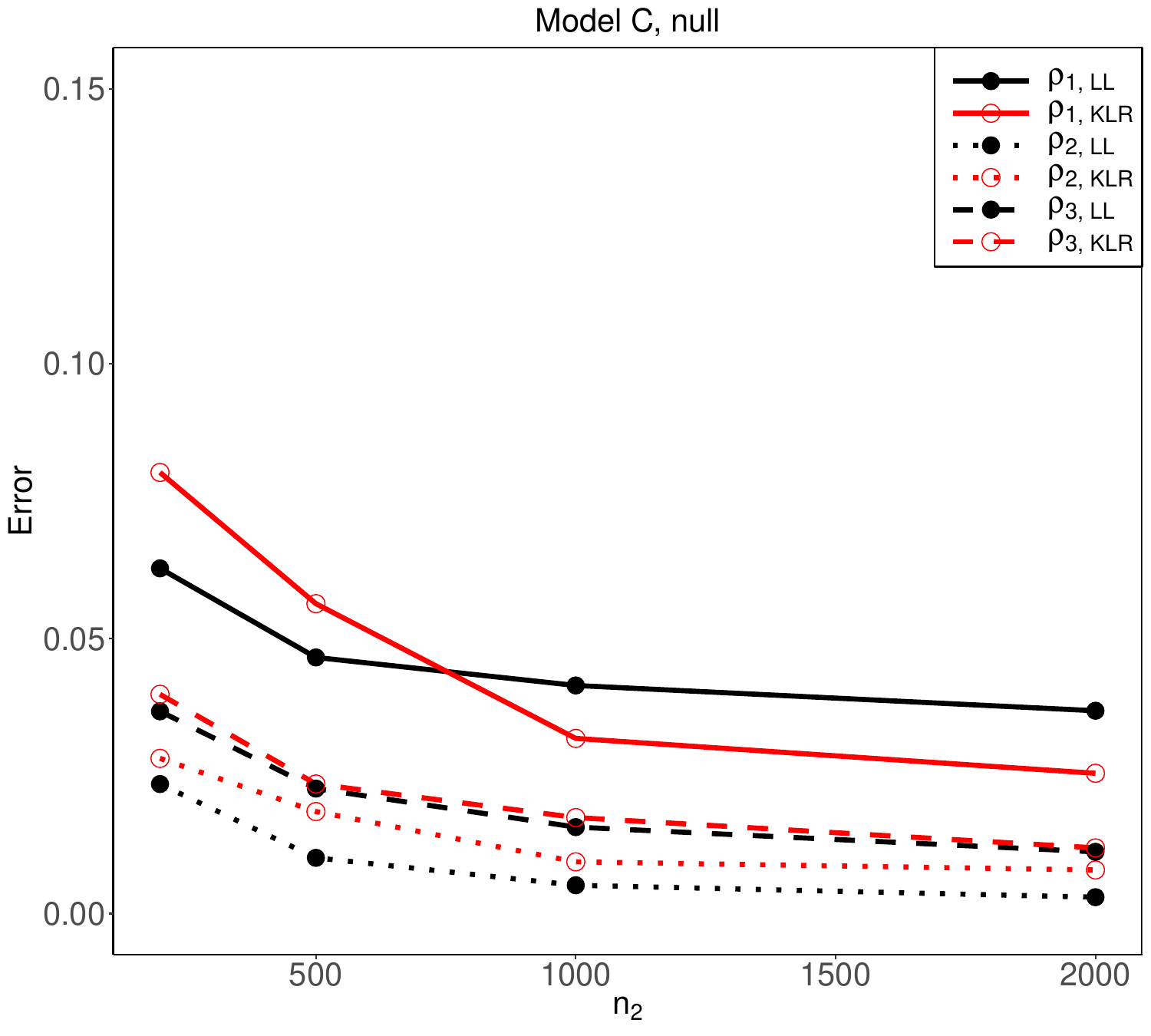} &
			\hspace{\thisgap}\includegraphics[width=\thiswidth]{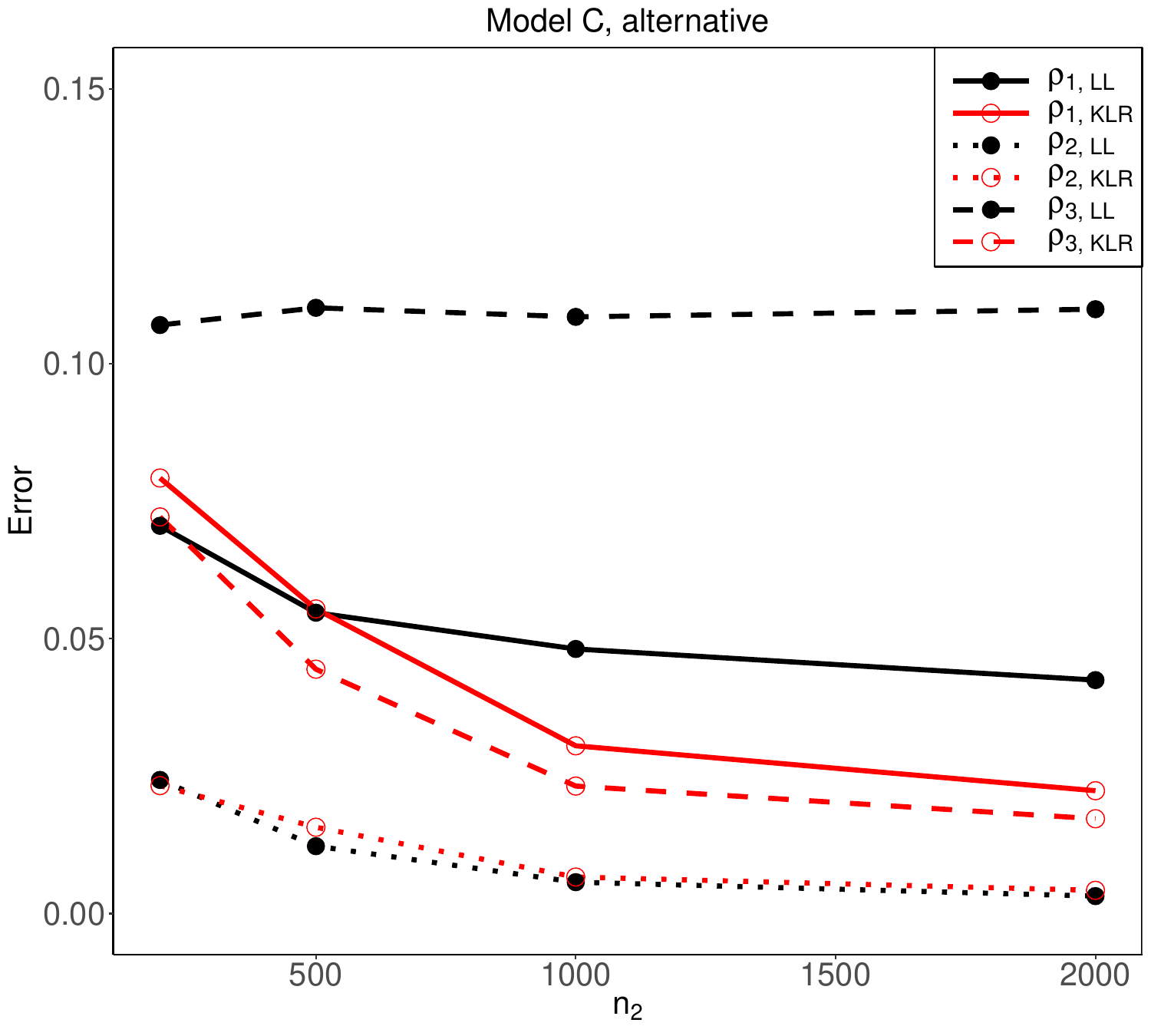}	
		\end{tabular}
		\caption{Error quantities under LL and KLR for all models with $p=500, s=5$ over 500 repetitions with $\alpha=0.05$ and the split ratio $r=0.5$.}
		\label{fig:error}
	\end{figure}

	\section{Auxiliary lemmas}\label{sec:proof_aux}
	
	In the rest of the Appendix we provide proofs of main results and related auxiliary lemmas.  The notation will mostly follow the main text, although some proofs involve their own notation.  Throughout the proofs, we use $C$ to denote a constant whose value only depends on $(P_1,P_2)$ but not the sample sizes $(n_1,n_2)$.  The value of $C$ may also change from line to line.
	
	Our first auxiliary lemma provides an analogous CDF transformation for discrete random variable.
	\begin{lemma}\label{lem:disc_unif}
		Let $V$ be a discrete random variable with finite support.  Let  $F(\cdot)$ be the CDF of $V$, with $F_-(\cdot)$ being its left limit. Let $p(\cdot)$ be the corresponding probability mass function.  Let $\zeta$ be an independent $U(0,1)$ random variable. Then
		$U=F_-(V)+\zeta p(V)\sim U(0,1)$.
	\end{lemma}
	\begin{proof}[Proof of \Cref{lem:disc_unif}]
		Let $v_1<...<v_s$ be the support of $V$.  Let $q_0=0$, $q_j=p(v_1)+...+p(v_j)$ for $j=1,...,s$.  Then $q_s=1$.
		For $1\le j\le s$, it is direct to verify that the density of $U$ on $(q_{j-1},q_j)$ is a constant.  Moreover, the length of this interval is $q_j-q_{j-1}=p(v_j)$, which is the same as the probability of $U$ falling in $(q_{j-1},q_j)$. So the density of $U$ on this interval is $1$. Since this holds for each $j$, we conclude that $U\sim U(0,1)$.
	\end{proof}
	
	The next lemma is useful in establishing separation between $H_0$ and $H_1$ when the correct $V$ function is used.  \Cref{lem:cov} provides a ``change-of-variable'' trick to simplify an integral involved in calculating $\mathbb E U$, the expected value of the conformal $p$-value given in \eqref{eq:U_weighted}. 
	\begin{lemma}\label{lem:cov}
		Let $(X_j,Y_j)\sim P_j$ be independent for $j=1,2$, and $(X_2',Y_2')$ be another independent draw from $P_2$.  Let $v(x,y)=\frac{f_1(y|x)}{f_2(y|x)}$ and $g(x)=f_{2,X}(x)/f_{1,X}(x)$. Define $V_1=v(X_1,Y_1)$, $V_2=v(X_2,Y_2)$, and $V_2'=v(X_2',Y_2')$. Let $\hat v:\mathcal X\times\mathcal Y\mapsto \mathbb R$ be an arbitrary non-random function and define $\hat V_1$, $\hat V_2$, $\hat V_2'$ similarly as $V_1,V_2,V_2'$ using $\hat v$. We have
		$$
		\mathbb E g(X_1)\indc(\hat V_1< \hat V_2) = \mathbb E V_2'\indc(\hat V_2'<\hat V_2)
		$$
		and
		$$
		\mathbb E g(X_1)\indc(\hat V_1\le \hat V_2) = \mathbb E V_2'\indc(\hat V_2'\le \hat V_2)\,.
		$$
	\end{lemma}
	\begin{proof}
		By definition,
		\begin{eqnarray*}
			&  & \mathbb E g(X_1)\mathbbm{1}(\hat V_1<\hat V_2)\nonumber\\
			&=& \mean \left\{ \frac{f_2(X_{1},Y_{1})}{f_1(X_{1},Y_{1})}\frac{f_1(Y_{1}|X_{1})}{f_2(Y_{1}|X_{1})}\indc(\hat V_1 < \hat V_2) \right\} \nonumber\\
			& =& \int f_2(x_2,y_2) \int_{\hat v(x_1, y_1) < \hat v(x_2, y_2)} f_2(x_1,y_1)v(x_1,y_1) dx_1dy_1dx_2dy_2  \nonumber\\
			& =& \mean \{V_2^{\prime}\indc(\hat V_2^{\prime} < \hat V_2)\}\,.
		\end{eqnarray*}
		The proof of the second equation is identical.
	\end{proof}

	\begin{lemma}\label{lem:tilde_sigma/hat_sigma}
		Let $\hat \sigma^2 = \hat \sigma_1^2 + n_{11}/(12n_{21})+\hat \sigma_2^2/4 - \hat \rho_{12} $ be defined as in \eqref{eq:hat_sigma}, and $\tilde \sigma^2 = \tilde \sigma_1^2 + n_{11}/(12n_{21})+\sigma_2^2/4 - \tilde \rho_{12}$ where $\tilde{\sigma}_1^2 = \mathrm{Var}_*[G_{11}\{1-\hat F_{1/2}(\hat V_{11})\}]$, $\sigma_2^2= \mathrm{Var}(G_{11})$ and $\tilde{\rho}_{12} =\mathrm{Cov}_*(G_{11}\{1-\hat F_{1/2}(\hat V_{11})\}, G_{11} )$.
		Then we have $\tilde \sigma/ \hat \sigma = 1+o_p(1)$ under Assumptions \ref{ass:moment}, \ref{ass:accuracy_h_g}.
	\end{lemma}
	
	\begin{proof}
		Recall that $\hat F$ is the CDF of $\hat v(X_2, Y_2)$ by treating $\hat v$ as fixed. Let $\hat F_n$ be the empirical CDF of $\hat v(X_2, Y_2)$ from the ranking sample. Because the ranking sample and $\hat{v}$ are independent, we have
		$$
		\|\hat F_n-\hat F\|_\infty=O_P(n_{21}^{-1/2})\,.
		$$
		
		Consider the estimate 
		$$
		\hat \sigma_1^2 = \frac{1}{n_{11}}\sum_{i}\hat G_{1i}^2 \left\{1-\hat{F}_{n,1/2}(\hat V_{1i})\right\}^2 - \left[\frac{1}{n_{11}} \sum_i \hat G_{1i}\{1-\hat F_{n,1/2}(\hat V_{1i})\} \right]^2.
		$$
		
		Then we have
		\begin{eqnarray*}
			& & \hat \sigma_1^2 - \tilde \sigma_1^2 \\
			& = & \frac{1}{n_{11}}\sum_{i}\hat G_{1i}^2 \left\{1-\hat{F}_{n,1/2}(\hat V_{1i})\right\}^2 - \mathbb{E}_*\left[G_{11}^2\left\{1-\hat F_{1/2}(\hat V_{11})\right\}^2\right] + \\
			& & \mathbb E_*^2 [G_{11}\{ 1-\hat F_{1/2}(\hat V_{11}) \}] - \left[\frac{1}{n_{11}} \sum_i \hat G_{1i}\{1-\hat F_{n,1/2}(\hat V_{1i})\} \right]^2 \\
			& = & I + II.
		\end{eqnarray*}
		To control the term I, we have
		\begin{eqnarray*}
			I & = & \frac{1}{n_{11}}\sum_i (\hat G_{1i}^2 - G_{1i}^2)\{ 1-\hat{F}_{n,1/2}(\hat V_{1i}) \}^2\\
			& & + \frac{1}{n_{11}} \sum_i G_{1i}^2 \left[ \{1-\hat F_{n,1/2}(\hat V_{1i}) \}^2 - \{ 1-\hat F_{1/2}(\hat V_{1i}) \}^2  \right] \\
			& &  +\frac{1}{n_{11}} \sum_i G_{1i}^2\{1-\hat F_{1/2}(\hat V_{1i})\}^2 - 
			\mathbb{E}_*\left[G_{11}^2\left\{1-\hat F_{1/2}(\hat V_{11})\right\}^2\right] \\
			& = & I_1 + I_2 + I_3.
		\end{eqnarray*}
		Due to the fact that $\{ 1-\hat{F}_{n,1/2}(\hat V_{1i}) \}^2 \le 1$ and by \Cref{ass:accuracy_h_g}(a), 
		$$ |I_1|  \le \frac{1}{n_{11}} \sum_i |G_{1i}^2 - \hat G_{1i}^2| = o_P(1) \,. $$
		To control the term $I_2$, we have
		$$ |I_2| \le \frac{2\|\hat{F}_n - \hat F\|_{\infty}}{n_{11}}\sum_i G_{1i}^2 = O_P(n_{21}^{-1/2}). $$
		Controlling the term $I_3$ follows from the weak law of large numbers (WLLN), which gives $ |I_3| = o_P(1). $
		Putting these pieces together, we obtain $I=o_P(1)$. In a similar way, we conclude that $II=o_P(1)$. 
		
		Analogously, we could establish that $\sigma_2^2 - \hat \sigma_2^2 = o_P(1)$ and $\tilde{\rho}_{12} - \hat \rho_{12}=o_P(1)$. We complete the proof by using the continuous mapping theorem.	
	\end{proof}
	
	\begin{lemma}\label{lem:sigma/hat_sigma}
		Let $\hat\sigma^2$ and $\tilde\sigma^2$ be defined as in \eqref{eq:hat_sigma} and \Cref{lem:tilde_sigma/hat_sigma}.
		Let $\sigma^2 = \sigma_1^2 + n_{11}/(12n_{21})+\sigma_2^2/4 - \rho_{12}$ be the ideal version of $\tilde\sigma^2$ using the true conditional density ratio $v$. If $|\mathbb E_*G_{11}(\hat D_{11}-D_{11})|=o_P(n_{11}^{-1/2})$, then we have $\sigma/ \hat \sigma = 1+o_p(1)$.
	\end{lemma}
	
	\begin{proof}
		According to \Cref{lem:tilde_sigma/hat_sigma} and Slutsky's theorem, it suffices to prove that $\sigma / \tilde \sigma=1+o_P(1)$.
		
		Note that 
		\begin{align*}
		\sigma_1^2 - \tilde \sigma_1^2 = & \mathbb E_*G_{11}^2\left[ \{1-F_{1/2}(V_{11})\}^2 - \{1-\hat F_{1/2}(\hat V_{11}) \}^2 \right] \\
		&+ \mathbb E_*^2 G_{11}\{1-\hat F_{1/2}(\hat V_{11}) \} - \mathbb E^2 G_{11} \{1-F_{1/2}(V_{11})\} \\
		= & I+II.
		\end{align*}
		
		To bound the terms $I$ and $II$, we have
		$$ |I| \le 2 \mathbb E_* G_{11}^2|F_{1/2}(V_{11})- \hat F_{1/2}(\hat V_{11})| \,, $$
		$$ |II| \le 2\mathbb E G_{11} \mathbb E_* G_{11}|F_{1/2}(V_{11})- \hat F_{1/2}(\hat V_{11})| \,. $$
		An application of H\"{o}lder's inequality implies that $|I|=o_P(1)$ and $II=o_P(1)$. Analogously, we obtain $\tilde \rho_{12} - \rho_{12}=o_P(1)$, which completes the proof.
	\end{proof}
	
	\begin{lemma}\label{lem:hat_U_U_prime}
		Let $U_j'=\left(n_{11}^{-1}\sum_{i=1}^{n_{11}}G_{1i}\hat D_{ij}\right)/\left(n_{11}^{-1}\sum_{i=1}^{n_{11}}G_{1i}\right)$, and 
		$$T'=\frac{\frac{1}{2}-\frac{1}{n_{21}}\sum_{j=1}^{n_{21}}U_j'}{\hat\sigma/\sqrt{n_{11}}}\,.$$
		Then,  under Assumptions \ref{ass:moment}, \ref{ass:accuracy_h_g}, 
		$$
		\hat T - T'=  \frac{\mathbb E_* (\hat G_{11}-G_{11})\hat D_{11}-\mathbb E_*(\hat G_{11}-G_{11})\mathbb E_* (G_{11} \hat D_{11})}{\hat\sigma/\sqrt{n_{11}}} + o_P(1)\,,
		$$
		where the expectation is taken over the ranking subsample while $\hat g$ and $\hat v$ are treated as fixed.
		Moreover, we have $T' \rightsquigarrow N(0,1)$ as $n_{11}\to \infty$ under $H_0$.
	\end{lemma}
	
	\begin{proof}[Proof of \Cref{lem:hat_U_U_prime}]
		
		Note that
		\begin{align*}
		\hat U_{j}-U_j' = \frac{\frac{1}{n_{11}}\sum_i (\hat G_{1i}-G_{1i})\hat D_{ij}}{\frac{1}{n_{11}}\sum_i \hat G_{1i}}+\frac{\frac{1}{n_{11}}\sum_{i} G_{1i}\hat D_{ij}}{\frac{1}{n_{11}}\sum_i G_{1i}}\left(\frac{\frac{1}{n_{11}}\sum_i G_{1i}}{\frac{1}{n_{11}}\sum_i \hat G_{1i}}-1\right).
		\end{align*}
		
		By law of large numbers, we have $n_{11}^{-1}\sum_i G_{1i} = 1 + o_P(1)$. Under Assumption \ref{ass:accuracy_h_g}(a), we obtain $n_{11}^{-1} \sum_i \hat G_{1i} = 1 + o_P(1)$ and $n_{11}^{-1} \sum_i (G_{1i}-\hat G_{1i}) = \mathbb{E}_*(G_{11}-\hat G_{11}) + o_P(n_{11}^{-1/2})$. Since $|\hat D_{ij}| \le 1$ and $|n_{21}^{-1}\sum_j \hat D_{ij}| \le 1$, we have $(n_{11}n_{21})^{-1}\sum_{i,j}(\hat G_{1i}-G_{1i})\hat D_{ij} = \mathbb E_*(\hat G_{11}-G_{11})\hat D_{11}+o_P(n_{11}^{-1/2})$ and $n_{11}^{-1} \sum_{i} G_{1i}\hat D_{ij} = \mathbb E_* G_{11} \hat D_{11} + O_P(n_{11}^{-1/2})$.

		Thus, by continuous mapping theorem,
		\begin{align*}
		\frac{1}{n_{21}}\sum_j(\hat U_{j}-U_j')=&\frac{\frac{1}{n_{11}n_{21}}\sum_{i,j} (\hat G_{1i}-G_{1i})\hat D_{ij}}{\frac{1}{n_{11}}\sum_i \hat G_{1i}}+\frac{\frac{1}{n_{11}n_{21}}\sum_{i,j} G_{1i}\hat D_{ij}}{\frac{1}{n_{11}}\sum_i G_{1i}}\left(\frac{\frac{1}{n_{11}}\sum_i G_{1i}}{\frac{1}{n_{11}}\sum_i \hat G_{1i}}-1\right)\\
		=&\left[\mathbb E_*(\hat G_{11}-G_{11})\hat D_{11}+o_P(n_{11}^{-1/2})\right](1+o_P(1))\\
		&+\left[\mathbb E_* G_{11}\hat D_{11}+O_P(n_{11}^{-1/2})\right]\left[\mathbb E_* (G_{11}-\hat G_{11})+o_P(n_{11}^{-1/2})\right](1+o_P(1)).
		\end{align*}
		
		Then we have
		\begin{align*}
		\hat T - T'= & \frac{\mathbb E_* (\hat G_{11}-G_{11})\hat D_{11}-\mathbb E_*(\hat G_{11}-G_{11})\mathbb E_* (G_{11} \hat D_{11})}{\hat\sigma/\sqrt{n_{11}}} + o_P(1)\,.
		\end{align*}
		
		Next, we prove $T' \rightsquigarrow N(0,1)$ as $n_{11}\to \infty$ under $H_0$. Define
		$$ T'' = \frac{\frac{1}{2}-\frac{1}{n_{21}}\sum_{j=1}^{n_{21}}U_j'}{\tilde\sigma/\sqrt{n_{11}}}\,,$$
		where $\tilde{\sigma}^2$ is defined as in \Cref{lem:tilde_sigma/hat_sigma}.
		
		Recall that $\hat{F}$ is the conditional CDF of $\hat V_{21}$ given $\hat{v}$, $\hat{F}_{-}$ its left limit, and $\hat{F}_\zeta = (1-\zeta)\hat F_-+\zeta \hat F$ for $\zeta\in[0,1]$.
		
		Using the marginal projection of a two sample U-statistic, we have
		\begin{equation}\label{eq:marginal_projection}
		\frac{1}{n_{11}n_{21}}\sum_{i,j} G_{1i}\hat D_{ij} =	\frac{1}{n_{11}}\sum_{i=1}^{n_{11}}\hat H_{1i}+\frac{1}{n_{21}}\sum_{j=1}^{n_{21}}\hat H_{2j}+\frac{1}{2}+R
		\end{equation}
		where
		\begin{align*}
		\hat H_{1i} = & \mathbb{E}_*\left[G_{1i}\hat D_{ij}|(X_{1i},Y_{1i})\right] -1/2 =  g(X_{1i})\{1-\hat F_{1/2}(\hat V_{1i})\} - 1/2\,,\\
		\hat H_{2j} = & \mathbb{E}_*\left[G_{1i}\hat D_{ij}|(X_{2j},Y_{2j}),\zeta_j\right] -1/2 =  \hat F_{\zeta_j}(\hat V_{2j}) - 1/2\,,
		\end{align*}
		and $R$ is a remainder term to make \eqref{eq:marginal_projection} hold, which satisfies
		\begin{align*}
		\mathbb E_* R^2 = \frac{1}{n_{11}^2 n_{21}^2} \sum_{i,j,i',j'}\mathbb E_* \tilde H_{ij}\tilde H_{i'j'}\,,
		\end{align*}
		where $\tilde H_{ij}=G_{1i}\hat D_{ij}-\hat H_{1i}-\hat H_{2j}-1/2$.
		By construction,  $\mathbb E_* \tilde H_{ij}\tilde H_{i'j'}$ is nonzero only if $(i,j)=(i',j')$, and $\mathbb E_* R^2$ reduces to
		$$
		\mathbb E_* R^2=\frac{1}{n_{11}^2n_{21}^2}\sum_{i,j} \mathbb E_* \tilde H_{ij}^2 = O((n_{11}n_{21})^{-1})\,.
		$$
		
		Now we can write the main part of the numerator of  $T''$ as
		$$
		\frac{1}{n_{21}}\sum_{j=1}^{n_{21}} U_j' = \frac{\frac{1}{n_{11}}\sum_{i=1}^{n_{11}}\hat H_{1i}+\frac{1}{n_{21}}\sum_{j=1}^{n_{21}}\hat H_{2j}+\frac{1}{2}+R}{\frac{1}{n_{11}}\sum_{i=1}^{n_{11}}G_{1i}} \,.
		$$
		Now apply the Lindeberg-Feller CLT to the triangular array $\{(\hat H_{1i},G_{1i}),1\le i\le n_{11}, \hat H_{2j}:1\le j\le n_{21}\}$ indexed by $n_{11}$, combining with the delta method, we have for any $\hat v$,
		$$  T'' | \hat v \rightsquigarrow N(0, 1) \,.$$
		Then, for any bounded continuous function $f$ and $Z\sim N(0,1)$, we have
		$ \mathbb E_* f(T'') \stackrel{a.s.}{\longrightarrow} \mathbb E f(Z)\,. $
		According to the bounded convergence theorem, 
		$ \mathbb E f(T'') \to \mathbb E f(Z)\,. $
		Thus, we conclude that $T'' \rightsquigarrow N(0,1)$ as $n_{11} \to \infty$.
		Combining with \Cref{lem:tilde_sigma/hat_sigma} and Slutsky's theorem yields $T' \rightsquigarrow N(0,1)$ as $n_{11} \rightarrow \infty$ under $H_0$.
	\end{proof}

	\section{Proofs of main results}\label{sec:proof_main}
	
	\begin{proof}[Proof of \Cref{lem:weighted_conf}]
		For part (a), the proof is a standard application of the Bayes theorem, and is implicitly given in \cite[equation (6) and Section 3.2]{tibshirani2019conformal}.  Here we provide the calculation for the readers' convenience.
		
		Recall that given $\tilde{\mathbf Z}$, we have $V_{n_{11}+1}=v(\tilde X_i,\tilde Y_i)$ if $\sigma(n_{11}+1)=i$. Thus
		\begin{equation}\label{eq:V_given_Z}
		(V_{n_{11}+1}|\tilde{\mathbf Z}) \sim \sum_{\sigma} P(\sigma|\tilde{\mathbf Z})\delta_{v(\tilde X_{\sigma(n_{11}+1)},\tilde Y_{\sigma(n_{11}+1)})} = \sum_{i=1}^{n_{11}+1} P\left[\sigma(n_{11}+1)=i|\tilde{\mathbf Z}\right]\delta_{v(\tilde X_{i},\tilde Y_{i})}\,.
		\end{equation}
		Using Bayes rule,
		\begin{align*}
		P\left[\sigma(n_{11}+1)=i|\tilde{\mathbf Z}\right]	=&  \frac{\sum_{\sigma(n_{11}+1)=i} \Pi_{l=1}^{n_{11}+1} f_{1,X}(\tilde X_l)f(\tilde Y_l|\tilde X_l) \frac{f_{2,X}(\tilde X_i)}{f_{1,X}(\tilde X_i)}}{\sum_{j=1}^{n_{11}+1}\sum_{\sigma(n_{11}+1)=j} \Pi_{l=1}^{n_{11}+1} f_{1,X}(\tilde X_l)f(\tilde Y_l|\tilde X_l) \frac{f_{2,X}(\tilde X_j)}{f_{1,X}(\tilde X_j)}}\\
		=& \frac{\frac{f_{2,X}(\tilde X_i)}{f_{1,X}(\tilde X_i)}}{\sum_{j=1}^{n_{11}+1} \frac{f_{2,X}(\tilde X_j)}{f_{1,X}(\tilde X_j)}}=p_i(\tilde{\mathbf Z})\,.
		\end{align*}
		As a result, \eqref{eq:V_given_Z} becomes
		\begin{align*}
		(V_{n_{11}+1}|\tilde{\mathbf Z}) \sim \sum_{i=1}^{n_{11}+1} p_i(\tilde{\mathbf Z})\delta_{v(\tilde X_{i},\tilde Y_{i})}=\sum_{i=1}^{n_{11}+1}p_i(\mathbf Z)\delta_{v(X_{i},Y_{i})}\,.
		\end{align*}
		
		Part (b) follows directly by combining part (a) with \Cref{lem:disc_unif}. 
		
		For part (c),
		we use notation $V_{1l} = v(X_{l}, Y_{l}) = f_1(Y_{l}|X_{l})/f_2(Y_{l}|X_{l})$, $l=1,\dots,n_{11}$, and $V_2 = v(X_{n_{11}+1}, Y_{n_{11}+1})=f_1(Y_{n_{11}+1}|X_{n_{11}+1})/f_2(Y_{n_{11}+1}|X_{n_{11}+1})$.  Recall that $g(x)=f_{2,X}(x)/f_{1,X}(x)$.
		
		For a given $X_2$, we consider 
		$$
		2\mathbb E_\zeta U = \frac{\sum_{l=1}^{n_{11}} g(X_{l})[\indc(V_{1l}<V_2)+\indc(V_{1l}\le V_2)]+g(X_{n_{11}+1})}{\sum_{l=1}^{n_{11}} g(X_{l})+g(X_{n_{11}+1})}\,,
		$$
		where $\mathbb E_\zeta$ denotes expectation taken only over $\zeta$, keeping everything else as given.
		Because $\mathbb E g(X_{l})=1$ for $1\le l\le n_{11}$, it can be directly verified from the strong law of large numbers on the iid random variables $X_{1},...,X_{n_{11}}$ that  
		$$2\mathbb E_{\zeta}U\rightarrow \mathbb E \left[g(X_{1})\mathbbm{1}(V_{11}<V_2)\bigg | X_{n_{11}+1},Y_{n_{11}+1}\right]+\mathbb E \left[g(X_{1})\mathbbm{1}(V_{11}\le V_2)\bigg | X_{n_{11}+1},Y_{n_{11}+1}\right]$$ a.e. over $(X_{1},...,X_{n_{11}},X_{n_{11}+1})$ as $n_{11}\rightarrow\infty$. 
		By construction we have $U\in [0,1]$. By the dominated convergence theorem, we have
		\begin{equation}\label{eq:limit_EU}
		2\mathbb E (U) \rightarrow \mathbb E g(X_{11})\mathbbm{1}(V_{11}<V_2)+\mathbb E g(X_{11})\mathbbm{1}(V_{11}\le V_2)\,.
		\end{equation}
		
		\Cref{lem:cov} implies that the right hand side of \eqref{eq:limit_EU} is (letting $V_2'$ be an iid copy of $V_2$)
		\begin{align*}
		&\mathbb E V_2' \indc(V_2'< V_2)+\mathbb E V_2'\indc(V_2'\le V_2)\\
		=&\mathbb E V_2 \indc(V_2< V_2') + 1-\mathbb E V_2'\indc(V_2'> V_2)\\
		=&1-\left[\mathbb E V_2'\indc(V_2'> V_2)-\mathbb E V_2\indc(V_2'>V_2)\right]\\
		=&1-\mathbb E (V_2'-V_2)\indc(V_2'>V_2)\\
		=&1-\frac{1}{2}\mathbb E|V_2'-V_2|\,,
		\end{align*} 
		which completes the proof.
	\end{proof}
	
	\begin{proof}[Proof of Theorem \ref{thm:null-shift}]
		Define 	$$T'=\frac{\frac{1}{2}-\frac{1}{n_{21}}\sum_{j=1}^{n_{21}}U_j'}{\hat\sigma/\sqrt{n_{11}}}\,,$$
		where $U'_j = \left(n_{11}^{-1}\sum_{i=1}^{n_{11}}G_{1i}\hat D_{ij}\right)/\left(n_{11}^{-1}\sum_{i=1}^{n_{11}}G_{1i}\right)$ and $\hat{D}_{ij} = \mathbbm{1}(\hat V_{1i} < \hat V_{2j}) + \zeta_j \mathbbm{1}(\hat V_{1i} = \hat V_{2j})$.
		Then \Cref{ass:accuracy_h_g}(b) and \Cref{lem:hat_U_U_prime} imply that $\hat T = T' + \hat T - T' \rightsquigarrow N(0,1)$ under $H_0$.
		
		Next we prove the asymptotic power under the alternative. Note that, letting $Z_{k1}=(X_{k1},Y_{k1})$ for $k=1,2$,
		\begin{eqnarray}\label{eq:var_U_U'}
		& & \mathrm{Var}_*\left\{\frac{1}{n_{11}n_{21}}\sum_{i,j} G_{1i}(D_{ij}-\hat D_{ij})\right\} \nonumber \\
		& = & \left(\frac{1}{n_{11}} \mathrm{Var_*}[\mathbb E_*\{G_{11}(D_{11}-\hat D_{11})| Z_{11} \}]  +  \frac{1}{n_{21}} \mathrm{Var}_*[\mathbb E_*\{G_{11}(D_{11}-\hat D_{11})| Z_{21}, \zeta_1 \}] \right) (1+o_P(1)) \nonumber \\
		& \le & \left[ \frac{C}{n_{11}} \mathrm{Var}_*\{G_{11}(D_{11} - \hat D_{11})\} \right](1+o_P(1)),
		\end{eqnarray}
		due to \Cref{ass:moment} and the boundedness of $\hat{F}_{\zeta}$ and $F_{\zeta}$ for any $\zeta \in [0,1]$.
		Thus, we have
		$$ \frac{1}{n_{21}}\sum_j (U_j-U_j') = \frac{\frac{1}{n_{11}n_{21}}\sum_{i,j} G_{1i}(D_{ij}-\hat D_{ij})}{\frac{1}{n_{11}}\sum_i G_{1i}} = \{\mathbb E_*G_{11}( D_{11}-\hat D_{11}) + O_P(1/\sqrt{n_{11}})\}(1+o_P(1)). $$
		Then,
		\begin{align}\label{eq:T'}
		T' = & \frac{\frac{1}{2}-\frac{1}{n_{21}}\sum_j U_j'}{\hat\sigma/\sqrt{n_{11}}} \nonumber\\
		=&\frac{\frac{1}{2}-\frac{1}{n_{21}}\sum_j U_j}{\hat\sigma/\sqrt{n_{11}}}+\frac{\frac{1}{n_{21}}\sum_j (U_j-U_j')}{\hat\sigma/\sqrt{n_{11}}} \nonumber\\
		=&\frac{\frac{1}{2}-\frac{1}{n_{21}}\sum_j U_j}{\hat\sigma/\sqrt{n_{11}}}-\frac{\mathbb E_* G_{11}(\hat D_{11}-D_{11}) }{\hat\sigma/\sqrt{n_{11}}}(1+o_P(1))+O_P(1) \nonumber\\
		= & \frac{\frac{1}{2} -\mathbb E G_{11}D_{11} }{\hat \sigma/\sqrt{n_{11}} } + \frac{\mathbb E G_{11}D_{11} - \frac{1}{n_{21}}\sum_j U_j }{\hat \sigma / \sqrt{n_{11}} } -\frac{\mathbb E_* G_{11}(\hat D_{11}-D_{11}) }{\hat\sigma/\sqrt{n_{11}}}(1+o_P(1))+O_P(1) \nonumber\\
		=&\frac{\frac{1}{4}\delta}{\hat\sigma/\sqrt{n_{11}}}+\frac{\mathbb E G_{11}D_{11} -\frac{1}{n_{21}}\sum_j U_j}{\sigma/\sqrt{n_{11}}}\frac{\sigma}{\hat\sigma}-\frac{\mathbb E_* G_{11}( \hat D_{11}-D_{11}) }{\hat\sigma/\sqrt{n_{11}}}(1+o_P(1))+O_P(1),
		\end{align}
		where $\delta = \mathbb E|v(X_2, Y_2)- v(X'_2, Y'_2)|$, and the last equality holds due to the fact that $\delta/4 = 1/2 - \mathbb EG_{11}D_{11}$ as explained in the proof of \Cref{lem:weighted_conf}(c).
		Note that 
		$$ \frac{\mathbb E G_{11}D_{11} -\frac{1}{n_{21}}\sum_j U_j}{\sigma/\sqrt{n_{11}}} \rightsquigarrow N(0,1) $$
		and $\sigma/ \hat \sigma = O_P(1)$.
		If there exists a constant $c>0$ such that 
		$$\mathbb P\left[\mathbb E_* G_{11}(\hat D_{11}-D_{11})<(1/4)\mathbb E|v(X_2,Y_2)-v(X_2',Y_2')|-c\right]\rightarrow 1\,,$$ where $(X_2,Y_2)$ and $(X_2',Y_2')$ are iid realizations from $P_2$,
		then we have $T' \to \infty$ in probability. Together with $\hat{T}-T' = o_P(1)$, we establish the desired result.
	\end{proof}
	
	\begin{proof}[Proof of Equation \eqref{eq:TV_link}]
		First realize that
		\begin{align*}
		&\mathbb E_{P_2} D_{\rm tv}(f_1(\cdot|X),f_2(\cdot|X))\\ = & \frac{1}{2}\int\int \left|f_1(y|x)-f_2(y|x)\right|dy f_2(x)dx\\ = & \frac{1}{2}\int\int \left|\frac{f_1(y|x)}{f_2(y|x)}-1\right|f_2(y|x)f_2(x)dy dx=\frac{1}{2}\mathbb E_{P_2}|v(X,Y)-1|\,.
		\end{align*}
		The lower bound follows from a conditional Jensen's inequality:
		$$
		\mathbb E|v(X,Y)-1|=\mathbb E\left|v(X,Y)-\mathbb E v(X',Y')\right|\le \mathbb E |v(X,Y)-v(X',Y')|\,,
		$$
		where $(X,Y)$, $(X',Y')$ are iid copies from $P_2$.
		
		The upper bound follows from triangle inequality:
		\begin{align*}
		&\mathbb E|v(X,Y)-v(X',Y')|=\mathbb E\left|v(X,Y)-1-[v(X',Y'-1)]\right|\\ 
		\le & \mathbb E|v(X,Y)-1|+\mathbb E|v(X',Y')-1|=2\mathbb E |v(X,Y)-1|\,.\qedhere
		\end{align*}	
	\end{proof}
	
	\begin{proof}[Proof of \Cref{pro:local_alt}]
		From \eqref{eq:var_U_U'}, we have
		$$ \mathrm{Var}_*\left\{\frac{1}{n_{11}n_{21}}\sum_{i,j} G_{1i}(D_{ij}-\hat D_{ij})\right\} \le \left[ \frac{C}{n_{11}} \mathrm{Var}_*\{G_{11}(D_{11} - \hat D_{11})\} \right](1+o_P(1)) \,. $$
		Since $|D_{11} - \hat D_{11}|\le 1$,
		$$ \mathrm{Var}_*\{G_{11}(D_{11} - \hat D_{11})\} \le \mathbb E_* G_{11}^2(D_{11}-\hat D_{11})^2 \le \mathbb E_* G_{11}^2|D_{11}-\hat D_{11}|. $$
		
		Now we study $\hat D_{11}-D_{11}$.
		Let $\xi = V_{21}-\hat V_{21} - (V_{11}-\hat V_{11})$. Then, for any $\epsilon>0$
		\begin{align*}
		&\left|\mathds{1}(\hat V_{11} < \hat V_{21}) - \mathds{1}(V_{11} < V_{21})\right|\\
		=&\mathds{1}(V_{11}+\xi<V_{21}\le V_{11},~\xi<0)+\mathds{1}(V_{11}< V_{21}\le V_{11}+\xi,~\xi>0)\\
		\le & \mathds{1}(V_{11}-|\xi|\le V_{21}\le V_{11}+|\xi|) \\
		\le & \mathds{1}(V_{11}-\epsilon\le V_{21}\le V_{11}+\epsilon) + \mathds{1}(|\xi|\ge \epsilon)\,.
		\end{align*}
		The same upper bound holds for $\mathds{1}(\hat V_{11} \le \hat V_{21}) - \mathds{1}(V_{11} \le V_{21})$ using the same argument, and we conclude that
		\begin{align*}
		|\hat D_{11}-D_{11}|= &\left| (1-\zeta_1)\left[\mathds{1}(\hat V_{11} < \hat V_{21}) - \mathds{1}(V_{11} < V_{21})\right] + \zeta_1\left[\mathds{1}(\hat V_{11} \le \hat V_{21}) - \mathds{1}(V_{11} \le V_{21})\right]\right| \\
		\le & (1-\zeta_1)\left|\mathds{1}(\hat V_{11} < \hat V_{21}) - \mathds{1}(V_{11} < V_{21})\right| + \zeta_1\left|\mathds{1}(\hat V_{11} \le \hat V_{21}) - \mathds{1}(V_{11} \le V_{21})\right| \\
		\le & \mathds{1}(V_{11}-\epsilon\le V_{21}\le V_{11}+\epsilon) + \mathds{1}(|\xi|\ge \epsilon)\,.	
		\end{align*}
		Therefore,
		\begin{align*}
		\mathbb E_* G_{11}^2|\hat D_{11}-D_{11}|\le & \mathbb E_* G_{11}^2\mathds{1}(V_{11}-\epsilon\le V_2\le V_{11}+\epsilon)
		+\mathbb E_* G_{11}^2 \mathds{1}(|\xi|\ge \epsilon)\\
		\le &  \mathbb E_* \left\{G_{11}^2\mathbb E_*\left[\mathds{1}(V_{11}-\epsilon\le V_2\le V_{11}+\epsilon)|X_{11},Y_{11}\right]\right\}+\mathbb E_* G_{11}^2 \mathds{1}(|\xi|\ge \epsilon)\\
		\le & 2C\epsilon \mathbb E G_{11}^2+\mathbb E_* G_{11}^2 \mathds{1}(|\xi|\ge \epsilon)\,,
		\end{align*}
		where the first term in the last inequality follows from the assumption of bounded density of $V_{21}$ with the constant $C$ being a finite upper bound of the density. Because $\xi=o_P(1)$ the support of $\mathds{1}(|\xi|\ge \epsilon)$ has vanishing probability measure.  The integrability of $G_{11}^2$ implies it is uniformly integrable. So we have $\mathbb E_* G_{11}^2 \mathds{1}(|\xi|\ge \epsilon)=o_P(1)$ for arbitrary $\epsilon>0$.  Therefore, we conclude that $\mathbb E_* G_{11}^2|\hat D_{11}-D_{11}|=o_P(1)$.
		Combining this variance bound with the mean bound $|\mathbb E_* G_{11}(\hat D_{11}-D_{11})|=o_P(1/\sqrt{n_{11}})$ assumed in the proposition, we have
		$$ \frac{1}{n_{21}}\sum_j (U_j-U_j') = \frac{\frac{1}{n_{11}n_{21}}\sum_{i,j} G_{1i}(D_{ij}-\hat D_{ij})}{\frac{1}{n_{11}}\sum_i G_{1i}} = o_P(1/\sqrt{n_{11}}). $$
		
		As in the proof of \Cref{thm:null-shift}, \Cref{ass:accuracy_h_g}(b) and \Cref{lem:hat_U_U_prime} imply that $\hat T - T' = o_P(1)$.
		According to \eqref{eq:T'} and \Cref{lem:sigma/hat_sigma}, we have
		\begin{eqnarray*}
			\hat T & = & \hat T - T' + T' \\
			& = & \frac{\frac{1}{4}\delta}{\hat\sigma/\sqrt{n_{11}}}+\frac{\mathbb E G_{11}D_{11} -\frac{1}{n_{21}}\sum_j U_j}{\sigma/\sqrt{n_{11}}}\frac{\sigma}{\hat\sigma} + \frac{\frac{1}{n_{21}}\sum_j (U_j-U_j')}{\hat\sigma/\sqrt{n_{11}}} + o_P(1) \\
			& = & \frac{\sqrt{n_{11}}\delta}{4\sigma}(1+o_P(1))+Z+o_P(1),
		\end{eqnarray*}
		where $Z\rightsquigarrow N(0,1)$ as $n_{11}\rightarrow\infty$. 
	\end{proof}

\end{document}